\documentclass[11pt]{article}

\usepackage{geometry}
\usepackage{amsmath,amssymb,amsthm}
\usepackage{enumerate}
\usepackage[colorlinks=true,citecolor=blue,linkcolor=blue,urlcolor=blue]{hyperref}
\usepackage{mathtools}
\usepackage{tikz}
\usepackage{url}

\usepackage{pifont}
\newcommand{\cmark}{\ding{51}}%
\newcommand{\xmark}{\ding{55}}%

\usetikzlibrary{positioning,chains,fit,shapes,calc}

\newcommand{\N}{\mathbb{N}}

\newcommand{\R}{\mathbb{R}}

\newcommand{\set}[1]{\{#1\}}
\newcommand{\sset}[2]{\{#1 \, : \, #2\}}

\renewcommand{\vec}[1]{\mathbf{#1}}

\newcommand{\msd}[1]{\rho(#1)}

\definecolor{myblue}{RGB}{80,80,160}
\definecolor{mygreen}{RGB}{80,160,80}

\usepackage{tabularx}
\newcolumntype{B}{>{\centering\arraybackslash}p{.95cm}}%
\newcolumntype{C}{>{\centering\arraybackslash}p{3cm}}%
\newcolumntype{D}{>{\raggedright\arraybackslash}p{3cm}}%
\newcolumntype{A}{>{\centering\arraybackslash}p{2cm}}%
\newcolumntype{E}{>{\centering\arraybackslash}p{2.5cm}}%

\usepackage{multirow}

\usepackage{thmtools,thm-restate}
\usepackage{cleveref}

\usepackage{ifthen}

\newtheorem{theorem}{Theorem}
\newtheorem{lemma}[theorem]{Lemma}
\newtheorem{corollary}[theorem]{Corollary}

\newtheorem{definition}[theorem]{Definition}
\newtheorem{example}[theorem]{Example}

\newtheorem{remark}[theorem]{Remark}

\begin{document}
\title{\bfseries
Topological Price of Anarchy Bounds for \\
Clustering Games on Networks%
\footnote{A preliminary version of this work appeared in the Proceedings of the 15th International Conference on Web and Internet Economics (WINE 2019) \cite{KleerS2019}.}%
}

\author{
Pieter Kleer \\ 
Max Planck Institute for Informatics\\  
Saarland Informatics Campus (SIC)\\
Saarbr\"ucken, Germany\\
\texttt{pkleer@mpi-inf.mpg.de}
\and
Guido Sch\"afer \\ 
Centrum Wiskunde \& Informatica (CWI) \\
Vrije Universiteit Amsterdam \\
Amsterdam, The Netherlands\\
\texttt{g.schaefer@cwi.nl}
}

\maketitle            

\begin{abstract}
We consider clustering games in which the players are embedded in a network and want to coordinate (or anti-coordinate) their strategy with their neighbors. The goal of a player is to choose a strategy that
maximizes her utility given the strategies of her neighbors. Recent studies show that even very basic variants of these games exhibit a large Price of Anarchy: A large inefficiency between the total utility generated in centralized outcomes and equilibrium outcomes in which players selfishly try to maximize their utility. Our main goal is to understand how structural properties of the network topology impact the inefficiency of these games. We derive \emph{topological bounds} on the Price of Anarchy for different classes of clustering games. These topological bounds provide a more informative assessment of the inefficiency of these games than the corresponding (worst-case) Price of Anarchy bounds. As one of our main results, we derive (tight) bounds on the Price of Anarchy for clustering games on Erd\H{o}s-R\'enyi random graphs (where every possible edge in the network is present with a fixed probability), which, depending on the graph density, stand in stark contrast to the known Price of Anarchy bounds.
\end{abstract}

\newpage
\section{Introduction}
Clustering games on networks constitute a class of strategic games in which the players are embedded in a network and want to coordinate (or anti-coordinate) their choices with their neighbors. These games capture several key characteristics encountered in applications such as opinion formation, technology adoption, information diffusion or virus spreading on various types of networks, e.g., the Internet, social networks and biological networks.

Different variants of clustering games have recently been studied intensively in the algorithmic game theory literature, both with respect to the existence and the inefficiency of equilibria, see, e.g., \cite{Anshelevich2014,Apt2015,Feldman2015,GM09,Gourves2010,Hoefer2007,KPR13,Rahn2015}.
Unfortunately, several of these studies reveal that the strategic choices of the players may lead to equilibrium outcomes that are highly inefficient. 
Arguably the most prominent notion to assess the inefficiency of equilibria is the \emph{Price of Anarchy (PoA)} \cite{Koutsoupias1999}, which refers to the worst-case ratio of the optimal social welfare and the social welfare of a (pure) Nash equilibrium. It is known that even the most basic clustering games exhibit a large (or even unbounded) Price of Anarchy (see below for details).
These negative results naturally trigger the following questions: Is this high inefficiency inevitable in clustering games on networks? Or, can we trace more precisely what causes a large inefficiency? These questions constitute the starting point of our investigations:
\begin{center}
\textit{Our main goal in this paper is to understand how structural properties of the \\
network topology impact the Price of Anarchy in clustering games.}
\end{center}

\medskip
\noindent In general, our idea is that a more fine-grained analysis may reveal topological parameters of the network which can be used to derive more accurate bounds on the Price of Anarchy; we term such bounds \emph{topological Price of Anarchy bounds}. Given the many applications of clustering games on different types of networks, our hope is that such topological bounds will be more informative than the corresponding worst-case bounds. 
Clearly, this hope is elusive for a number of fundamental games on networks whose inefficiency is known to be \emph{independent} of the network topology, the most prominent example being selfish routing games studied in the seminal work by Rougharden and Tardos \cite{RT02}.
But, in contrast to these games, clustering games exhibit a strong \emph{locality property} induced by the network structure, i.e., the utility of each player is affected only by the choices of her direct neighbors in the network. 
This observation 
also motivates our choice of quantifying the inefficiency by means of topological parameters (rather than other parameters of the game).

In this paper, we derive topological bounds on the Price of Anarchy for different classes of clustering games. Our bounds reveal that the Price of Anarchy depends on different topological parameters in the case of symmetric and asymmetric strategy sets of the players and, depending on these parameters, stand in stark contrast to the known worst case bounds.
As one of our primary benchmarks, we use Erd\H{o}s-R\'enyi random graphs \cite{Gilbert1959} to obtain a precise understanding of how these parameters affect the Price of Anarchy. More specifically, we show that the Price of Anarchy of clustering games on random graphs, depending on the graph density, improves significantly over the worst case bounds.  
To the best of our knowledge, this is also the first work that addresses the inefficiency of equilibria on random graphs. (We note that Valiant and Roughgarden~\cite{Roughgarden2010} study Braess' paradox in large random graphs; see Section \ref{sec:rel_work}.)

We note that the applicability of our topological Price of Anarchy bounds is not limited to the class of Erd\H{o}s-R\'enyi random graphs. The main reason for using these graphs is that their structural properties are well-understood. In particular, our topological bounds can be applied to any graph class of interest (as long as certain structural properties are well-understood).

Apart from our topological Price of Anarchy bounds, we also give a complete characterization of what type of \emph{distribution rules}, that determine how utility generated by two adjacent players in the network is split when they (anti-)coordinate, guarantee the convergence of best-response dynamics in symmetric clustering games. 

Altogether, our results give a more fine-grained view on clustering and coordination games.

\subsection{Our Clustering Games}
We study a generalization of the unifying model of \emph{clustering games} introduced by Feldman and Friedler \cite{Feldman2015}:
We are given an undirected graph $G = (V, E)$ on $n = |V|$ nodes whose edge set $E = E_c \cup E_a$ is partitioned into a set of \emph{coordination} edges $E_c$ and a set of \emph{anti-coordination} edges $E_a$. (The game is called a \emph{coordination game} if all edges are coordination edges and an \emph{anti-coordination game} (or \emph{cut game}) if all edges are anti-coordination edges.)
Further, we are given a set $[c] = \set{1, \dots, c}$ of $c > 1$ colors and edge-weights $w : E \rightarrow \R_{\geq 0}$. (In this paper, we use $[k]$ to denote the set $\set{1, \dots, k}$ for a given integer $k \ge 1$.)
Each node $i$ corresponds to a player who chooses a color $s_i$ from her color set $S_i \subseteq [c]$.
We say that the game is \emph{symmetric} if $S_i = [c]$ for all $i \in V$ and \emph{asymmetric} otherwise. 
An edge $e = \set{i, j} \in E$ is \emph{satisfied} if it is a coordination edge and both $i$ and $j$ choose the same color, or if it is an anti-coordination edge and $i$ and $j$ choose different colors. 
The goal of player $i$ is to choose a color $s_i \in S_i$ such that the weight of all satisfied edges incident to $i$ is maximized.

We consider a generalization of these games by incorporating additionally: (i) individual player preferences (as in \cite{Rahn2015}), and (ii) different distribution rules (as in \cite{Anshelevich2014}): We assume that each player $i$ has a \emph{preference function} $q_i: S_i \rightarrow \mathbb{R}_{\ge 0}$ which encodes her preferences over the colors in $S_i$. Further, player $i$ has a \emph{split parameter} $\alpha_{ij} \ge 0$ for every incident edge $e = \set{i,j}$ which determines the share she obtains from $e$: if $e$ is satisfied then $i$ obtains a proportion of $\alpha_{ij}/(\alpha_{ij} + \alpha_{ji})$ of the weight $w_e$ of $e$. 
The utility $u_i(s)$ of player $i$ with respect to strategy profile $s = (s_1, \dots, s_n)$ 
is then the sum of the individual preference $q_i(s_i)$ and the total share of all satisfied edges incident to $i$. 
We consider the standard utilitarian \emph{social welfare} objective $u(s) = \sum_i u_i(s)$.

We use $\bar\alpha_e$ to denote the \emph{disparity} of an edge $e = \set{i, j}$, defined as $\bar{\alpha}_{e} =  \max\set{\alpha_{ij}/\alpha_{ji}, \alpha_{ji}/\alpha_{ij}}$, and let $\bar{\alpha} = \max_{e \in E} \bar{\alpha}_{e}$ refer to the maximum disparity of all edges. We say that the game has the \emph{equal-split distribution rule} if $\bar\alpha = {1}$ (equivalently, $\alpha_{ij} = \alpha_{ji}$ for all $\set{i, j} \in E$).

Our clustering games generalize several other strategic games, which were studied extensively in the literature before, such as \emph{max cut games} and \emph{not-all-equal satisfiability games} \cite{GM09}, \emph{max $k$-cut games} \cite{Gourves2010}, \emph{coordination games} \cite{Apt2015}, \emph{clustering games} \cite{Feldman2015} and \emph{anti-coordination games} \cite{KPR13}. In turn, in Appendix \ref{app:extensions} we provide some natural generalizations of our clustering games. However, we argue that the results obtained in this work do not carry over to those more general settings.

\subsection{Our Contributions}
We derive results for symmetric and asymmetric clustering games. 
We elaborate on our main findings for symmetric clustering games only below; our results for the asymmetric case are detailed in Section~\ref{sec:ext}. 
An overview of the bounds derived in this paper is given in Table~\ref{tab:overview}.

\begin{table}[t]
\begin{center}
\begin{footnotesize}
\begin{tabular}
{|D|BBB|C@{\ \ }A|E|}
\multicolumn{7}{c}{\bf SYMMETRIC CLUSTERING GAMES} \\
\multicolumn{7}{c}{} \\[-2ex]
\hline
\bf Graph topology & \bf Coord. & \bf Indiv. & \bf Distr. & \multicolumn{2}{c|}{\textbf{Topological PoA}} & \textbf{PoA} \\
\bf  & \bf only & \bf pref. & \bf $\alpha$ & \multicolumn{2}{c|}{(our bounds)} & (prev.~work)\\
\hline 
\hline 
arbitrary & \xmark & \cmark & $+$ & $1 + \left(1 + \bar{\alpha} \right) \msd{G}$ & (Thm.~\ref{thm:poa})& \multirow[m]{6}{*}{$c$ \cite{Anshelevich2014,Feldman2015}} \\
planar & \xmark & \cmark & $+$ & $\leq 4 + 3\bar\alpha$ & (Cor.~\ref{cor:planar}) & \\
arbitrary & \xmark & \cmark & $\mathbf{1}$ & $1 +2 \msd{G}$ & (Cor.~\ref{cor:poa_unweighted}) & \\
arbitrary & \xmark & \cmark & $\mathbf{1}$ & $\leq 5 + 2 \msd{G_c}$ & (Thm.~\ref{thm:poa_clustering}) & \\
sparse random & \cmark & \cmark & $\mathbf{1}$  & $\Theta(1)$ & (Cor.~\ref{cor:sparse}) & \\
dense random & \cmark & \xmark & $\mathbf{1}$  & $\Omega(c)$ &  (Thm.~\ref{thm:random_graph}) & \\
\hline
\end{tabular}

\medskip
\begin{tabular}
{|D|BBB|C@{\ \ }A|E|}
\multicolumn{7}{c}{\bf ASYMMETRIC COORDINATION GAMES} \\
\multicolumn{7}{c}{} \\[-2ex]
\hline
\bf Graph topology & \bf Coord. & \bf Indiv. & \bf Distr. & \multicolumn{2}{c|}{\textbf{$(\epsilon, k)$-Topological PoA}} & \textbf{$(\epsilon, k)$-PoA}\\
\bf  & \bf only & \bf pref. & $\alpha$ & \multicolumn{2}{c|}{(our bounds)} & (prev.~work) \\
\hline 
\hline
arbitrary & \cmark & \xmark & $\mathbf{1}$ & $\le 2\epsilon \Delta(G)$ & (Thm.~\ref{lem:max_degree}) & \multirow[b]{2}{*}{$\le 2\epsilon\frac{n-1}{k-1}$} \\
arbitrary & \cmark & \xmark & $\mathbf{1}$ & $\ge \epsilon (\frac{\Delta(G)}{k-1} - 1)$ & (Thm.~\ref{lem:max_degree}) & \\
dense random & \cmark & \xmark & $\mathbf{1}$ & $\Omega(\epsilon n)$ & & 
\multirow[t]{1}{*}{$\ge 2\epsilon\frac{n-k}{k-1} + 1$} \\
sparse random & \cmark & \xmark & $\mathbf{1}$ & $\Theta\big( \frac{\epsilon\ln(n)}{\ln\ln(n)}\big)$ & (Thm.~\ref{thm:max_degree}) & \multirow[t]{2}{*}{\cite{Rahn2015}}  \\
\ + common color & \cmark & \xmark & $\mathbf{1}$ & $O(1)$ & (Thm.~\ref{thm:average_degree}) & \multirow[t]{2}{*}{} \\
\hline
\end{tabular}

\bigskip
\caption{
Overview of our topological Price of Anarchy bounds for symmetric and asymmetric clustering games. A ``$+$'' or ``$\mathbf{1}$'' in the column ``distr. $\alpha$'' indicates whether the distribution rule $\alpha$ is positive or equal-split, respectively. 
$\bar\alpha$ is the maximum disparity, and $c$ is the number of colors. 
$\msd{G}$ and $\Delta(G)$ refer to the maximum subgraph density and maximum degree of $G$, respectively. 
The stated bounds for random graphs hold with high probability.
}
\label{tab:overview}
\end{footnotesize}
\end{center}
\end{table}

\medskip\noindent
\textbf{1. Topological Price of Anarchy Bound.} 
We show that the Price of Anarchy for symmetric clustering games is bounded as a function of the \emph{maximum subgraph density} of $G$ which is defined as $\msd{G} = \max_{S \subseteq V} \{|E[S]|/|S|\}$, where $|E[S]|$ is the number of edges in the subgraph induced by $S$.
More specifically, we prove that $\text{PoA} \le 1 + (1 + \bar{\alpha}) \msd{G}$ and that this bound is tight (even for coordination games).
Using this topological bound, we are able to show that the Price of Anarchy is at most $4+3\bar\alpha$ for clustering games on planar graphs and $1+2\msd{G}$ for coordination games with equal-split distribution rule. 
We also derive a (qualitatively) refined bound of $\text{PoA} \le 5 + 2\msd{G[E_c]}$ for clustering games with equal-split distribution rule which reveals that the maximum subgraph density with respect to the graph $G[E_c]$ (or simply $G_c$) induced by the \emph{coordination edges $E_c$ only} is the crucial topological parameter determining the Price of Anarchy.

These bounds provide more refined insights than the known (tight) bound of $\text{PoA} \le c$ (number of colors) on the Price of Anarchy for (i) symmetric coordination games with individual preferences and arbitrary distribution rule \cite{Anshelevich2014}, and (ii) clustering games without individual preferences and equal-split distribution rule \cite{Feldman2015} (both being special cases of our model).
An important point to notice here is that this bound indicates that the Price of Anarchy is unbounded if the number of colors $c = c(n)$ grows as a function of $n$. In contrast, our topological bounds are independent of $c$ and are thus particularly useful when this number is large (while the maximum subgraph density is small). 
Moreover, our refined bound of $5 + 2\msd{G[E_c]}$ mentioned above provides a nice qualitative bridge between the facts that for max-cut (or anti-coordination) games the price of anarchy is known to be constant, whereas for coordination games the price of anarchy might grow large.

\medskip\noindent
\textbf{2. Price of Anarchy for Random Coordination Games.} 
We derive the first price of anarchy bounds for coordination games on random graphs. 
We focus on the \emph{Erd\H{o}s-R\'enyi random graph model}  \cite{Gilbert1959} (also known as $G(n,p)$), where each graph consists of $n$ nodes and every edge is present (independently) with probability $p \in [0,1]$. 
More specifically, we show that the Price of Anarchy is constant (with high probability) for coordination games on sparse random graphs (i.e., $p = d/n$ for some constant $d > 0$) with equal-split distribution rule. 
In contrast, we show that the Price of Anarchy remains $\Omega(c)$ (with high probability) for dense random graphs (i.e., $p = d$ for some constant $0 < d \leq 1$).

Note that our constant bound on the Price of Anarchy for sparse random graphs stands in stark contrast to the deterministic bound of $\text{PoA} = c$ \cite{Anshelevich2014,Feldman2015} (which could increase with the size of the network). On the other hand, our bound for dense random graphs reveals that we cannot significantly improve upon this bound through randomization of the graph topology. 

It is worth mentioning that all our results for random graphs hold against an \emph{adaptive adversary} who can fix the input of the clustering game \emph{knowing} the realization of the random graph.
To obtain these results, we need to exploit some deep probabilistic results on the maximum subgraph density and the existence of perfect matchings in random graphs.

\medskip\noindent
\textbf{3. Convergence of Best-Response Dynamics.} 
In general, pure Nash equilibria are not guaranteed to exist for clustering games with \emph{arbitrary} distribution rules $\alpha$, even if the game is symmetric (see, e.g., \cite{Anshelevich2014}). While some sufficient conditions for the existence of pure Nash equilibria, or, the convergence of best-response dynamics (see also \cite{Anshelevich2014}) are known, a complete characterization is elusive so far. 

In this work, we instead obtain a complete characterization of the class of distribution rules which guarantee the convergence of best-response dynamics in clustering games on a fixed network topology. We prove that best-response dynamics converge if and only if $\alpha$ is a \emph{generalized weighted Shapley distribution rule} (Theorem \ref{thm:best_response_2}). Our proof relies on the fact that there needs to be some form  of \emph{cyclic consistency} similar to the one used in \cite{Gopalakrishnan2013}. 
In fact, our characterization results regarding the existence of pure Nash equilibria and convergence of best-response dynamics are conceptually similar to the work of Chen et al.~\cite{Chen2010} and Gopalakrishnan et al.~\cite{Gopalakrishnan2013} (see Section \ref{sec:existence} for more details). 

Prior to our work, the existence of pure Nash equilibria was known for certain special cases of coordination games only, namely for symmetric coordination games with individual preferences and $c = 2$ \cite{Anshelevich2014}, and for symmetric coordination games without individual preferences \cite{Feldman2015}. To the best of our knowledge, this is the first characterization of distribution rules in terms of best-response dynamics, which, in particular, applies to the settings in which pure Nash equilibria are guaranteed to exist for every distribution rule \cite{Anshelevich2014,Feldman2015}.

\subsection{Related Work}
\label{sec:rel_work}
The literature on clustering and coordination games is vast; we only include references relevant to our model here. The proposed model above is a mixture of (special cases of) existing models in \cite{Anshelevich2014,Apt2015,Feldman2015,Rahn2015}. 

Anshelevich and Sekar \cite{Anshelevich2014} consider symmetric coordination games with individual preferences and (general) distribution rules. They show existence of \emph{$\epsilon$-approximate $k$-strong equilibria}, $(\epsilon,k)$-equilibria for short, for various combinations; in particular, $(2,k)$-equilibria always exist for any $k$. 
Moreover, they show that the number of colors $c$ is an upper bound on the PoA. 
Apt et al. \cite{Apt2015} study asymmetric coordination games with unit weights, zero individual preferences, and equal-split distribution rules. They derive an almost complete picture of the existence of $(1, k)$-equilibria for different values of $c$. 
Feldman and Friedler \cite{Feldman2015} introduce a unified framework (as introduced above) for studying the (strong) Price of Anarchy in clustering games with individual preferences set to zero and equal-split distribution rules. 
In particular, they show that the number of colors is an upper bound on the PoA and that $2(n-1)/(k-1)$ is an upper bound on the $(1,k)$-PoA.
Rahn and Sch\"afer \cite{Rahn2015} consider the more general setting of polymatrix coordination games with equal-split distribution rule, of which our asymmetric coordination games with individual preferences are a special case. They show a bound of $2\epsilon(n-1)/(k-1)$ on the $(\epsilon,k)$-PoA and that an $(\epsilon,k)$-equilibrium is guaranteed to exist for any $\epsilon \geq 2$ and any $k$. 

There is also a vast literature on different variants of anti-coordination (or cut) games, see, e.g., \cite{Gourves2010,Hoefer2007} and the references therein, which are also captured by our clustering games. In a recent paper, Carosi and Gianpiero \cite{Carosi2018} consider so-called \emph{$k$-coloring games}. 
Moreover, clustering and coordination games were also studied on directed graphs \cite{Apt2015,Carosi2017}. Finally, certain coordination and clustering games can be seen as special cases of hedonic games \cite{Dreze1980}; we refer the reader to \cite{Bilo2018} for, in particular, a survey of recent literature on (fractional) hedonic games. 
Identifying topological inefficiency bounds for these type of games, as well as for clustering games on directed graphs, could be an interesting direction for future work. (Our results do not seem to extend to clustering games on directed graphs. One could model a directed edge $e = (i,j)$ by setting $\alpha_{ij} = 0$ and $\alpha_{ji} > 0$. E.g., Theorem \ref{thm:poa} does not apply then as $\bar{\alpha} = \infty$ in this case.)

Regarding the study of the inefficiency of equilibria on random graphs, closest to our work seems to be the work by \cite{Roughgarden2010}. They study the Braess paradox on large Erd\H{o}s-R\'enyi random graphs and show that for certain settings the Braess paradox occurs with high probability as the size of the network grows large. The study of randomness in games has also received some attention in other settings, see, e.g., \cite{Amiet2019,Barany2007}. These are mostly settings with small strategy sets and random utility functions, and are not comparable with ours.

In the case of equal-split distribution rules, our clustering games can also be modelled as exact potential (or congestion) games \cite{Rosenthal1973}. The  inefficiency of pure Nash equilibria in these games has received a lot of attention, see, e.g., \cite{Christodoulou2005b,Christodoulou2005,Aland2006,Caragiannis2011,Kleer2017EC,Kleer2019tcs} and references therein. However, none of these results are applicable to the clustering games considered in this work. 
Finally, our games are also a special case of so-called \emph{distributed welfare games} as studied, e.g., by Marden and Wierman \cite{Marden2013}.

\section{Preliminaries} \label{sec:pre}
As introduced above, an instance of a \emph{clustering game} $\Gamma = (G,c,(S_i),(\alpha_{ij}),w,q)$ is given by: 
\begin{itemize}
\item an undirected graph $G = (V,E)$, where the set of edges $E = E_c \cup E_a$ is partitioned into coordination edges $E_c$ and anti-coordination edges $E_a$; 
\item a subset $S_i \subseteq [c]$ of colors available to player $i \in V$;
\item a split parameter $\alpha_{ij} \ge 0$ for every player $i \in V$ and incident edge $\set{i, j} \in E$;
\item a weight function $w : E \rightarrow \R_{\geq 0}$ on the edges;
\item a vector $q = (q_i)_{i \in V}$ of individual preference functions $q_i : S_i \rightarrow \R_{\geq 0}$.
\end{itemize}
Whenever we refer to a \emph{clustering game} below, we assume that all of the above input parameters are non-trivial; we specify the respective restrictions otherwise.

Each node $i \in V$ corresponds to a player whose goal is to choose a color $s_i \in S_i$ from the set of colors available to her to maximize her utility
$$
u_i(s) = q_i(s_i) + \sum_{\{i,j\} \in E_c: s_i = s_j} \frac{\alpha_{ij}}{ \alpha_{ij} + \alpha_{ji}} w_{ij}  + \sum_{\{i,j\} \in E_a: s_i \neq s_j} \frac{\alpha_{ij}}{ \alpha_{ij} + \alpha_{ji}} w_{ij}.
$$

We call $\alpha = (\alpha_{ij}) \ge 0$ a \emph{distribution rule}. We assume that $\alpha$ satisfies $\alpha_{ij} + \alpha_{ji} > 0$ for every edge $e = \{i,j\} \in E$; in particular, not both $i$ and $j$ have a zero split for edge $e$. 
We say that $\alpha$ is \emph{positive} if $\alpha_{ij} > 0$ and $\alpha_{ji} > 0$ for all $e = \{i,j\} \in E$; we also write $\alpha > \mathbf{0}$.
Further, $\alpha$ is called the \emph{equal-split} distribution rule if $\alpha_{ij} = \alpha_{ji}$ for all $e = \{i,j\} \in E$; we also indicate this by $\alpha = \mathbf{1}$.
The \emph{disparity} of an edge $e = \set{i, j}$ is defined as $\bar{\alpha}_{e} =  \max\set{\alpha_{ij}/\alpha_{ji}, \alpha_{ji}/\alpha_{ij}}$ and we use $\bar{\alpha} = \max_{e \in E} \bar{\alpha}_{e}$ to denote the maximum disparity. 

We say that the clustering game is \emph{symmetric} if $S_i = \set{1, \dots, c}$ for every player $i \in V$ and \emph{asymmetric} otherwise. If we focus on symmetric clustering games, we omit the explicit reference of the strategy sets $(S_i)$ with $S_i = [c]$. 
A clustering game is called a \emph{coordination game} if $E_a = \emptyset$ and an \emph{anti-coordination game} (or \emph{cut game}) if $E_c = \emptyset$. We use $n = |V|$ to refer to the number of players. 

A strategy profile $s = (s_1, \dots, s_n) \in \times_{i \in V} S_i$ is an \emph{$\epsilon$-approximate $k$-strong equilibrium} with $\epsilon \ge 1$ and $k \in [n]$, or \emph{$(\epsilon, k)$-equilibria} for short, if for every set of players $K \subseteq V$ with $|K| \le k$ and every deviation $s'_K = (s'_i)_{i \in K}$, there is at least one player $j \in K$ such that $\epsilon \cdot u_j(s) \geq  u_j(s_{-K},s'_K)$. That is, for any joint deviation of the players in $K$ from strategy profile $s$, there is at least one player that cannot improve her utility by more than a factor $\epsilon$. 

Let $\text{$(\epsilon, k)$-NE}(\Gamma)$ be the set of all $(\epsilon, k)$-equilibria of a game $\Gamma$. 
The \emph{$(\epsilon,k)$-Price of Anarchy} of $\Gamma$ is then defined as
$$
(\epsilon,k)\text{-PoA}(\Gamma) = \max_{s \in \text{$(\epsilon, k)$-NE}(\Gamma)} \frac{u(s^*)}{u(s)},
$$
where $s^*$ a strategy profile maximizing the \emph{social welfare} objective $u(s) = \sum_{i \in V} u_i(s)$.
For a class of clustering games $\mathcal{G}$ the \emph{$(\epsilon,k)$-Price of Anarchy} is given by $(\epsilon,k)\text{-PoA}(\mathcal{G}) = \sup_{\Gamma \in \mathcal{G}} \ (\epsilon,k)\text{-PoA}(\Gamma)$.
We only consider pairs $(\epsilon,k)$ for which $\text{$(\epsilon, k)$-NE}(\Gamma) \neq \emptyset$ for all $\Gamma \in \mathcal{G}$. When $\epsilon = 0$ and $k = 1$ we simply write PoA$(\cdot)$ instead of $(1,1)$-PoA$(\cdot)$.

\subsection{Random Clustering Games}
In our probabilistic framework to study the Price of Anarchy of random clustering games, we use the well-known \emph{Erd\H{o}s-R\'enyi random graph model} \cite{Gilbert1959}, denoted by $G(n, p)$: There are $n$ nodes and every (undirected) edge is present (independently) with probability $p = p(n) \in [0,1]$. (Although this model was first introduced by Gilbert, it is often referred to as the \emph{Erd\H{o}s-R\'enyi random graph model}.)
We say that a random graph is \emph{sparse} if $p = d/n$ for some constant $d > 0$, and it is \emph{dense} if $p = d$ for some constant $0 < d < 1$. In this paper, we focus on random graph instances with equal-split distributions rules. (Some of our results naturally extend to more general distribution rules, but we omit the (technical) details here because they do not provide additional insights.)

Fix some probability $p = p(n) \in [0,1]$ and let $\beta = \beta(n,c(n))$ be a given function.
Define $\mathcal{G}_{G_n}$ as the set of all clustering games on random graph $G_n \sim G(n,p)$.
We say that the \emph{Price of Anarchy for random clustering games 
is at most $\beta$ with high probability ($\text{PoA}(\mathcal{G}_{G_n}) \leq \beta$, for short)} if
$
\mathbb{P}_{G_n \sim G(n,p)}\{\text{PoA} \left(\mathcal{G}_{G_n}\right) \leq \beta\} 
\geq 1 - o(1). 
$
We use a similar definition if we want to lower bound the Price of Anarchy. 
Finally, for a constant $\beta$ (independent of $n$ and $c$)
we say that the \emph{Price of Anarchy for random clustering games is $\beta$ with high probability ($\text{PoA}(\mathcal{G}_{G_n}) \rightarrow \beta$, for short)} if for all $\varepsilon > 0$
$
\mathbb{P}_{G_n \sim G(n,p)} \left\{|\text{PoA} \left(\mathcal{G}_{G_n}\right) - \beta| \le \varepsilon \right\} \ge 1 - o(1).
$
All our results for clustering games on random graphs hold with high probability.

\subsection{Shapley Distribution Rules}

We adapt the definition of Shapley distribution rules for resource allocation games  \cite{Gopalakrishnan2013} to our setting. 
A distribution rule $\alpha$ corresponds to a \emph{generalized weighted Shapley distribution rule} if and only if there exists a permutation $\sigma$ of the players in $V$ and weight vector $\gamma \in \R_{\ge 0}^V$ such that the following two conditions are satisfied for every edge $e = \{i,j\}$:
\begin{enumerate}[(i)]
\item If $\alpha_{ij} = 0$, then $\sigma(i) < \sigma(j)$. 
\item If $\alpha_{ij} > 0$, then 
$
\frac{\alpha_{ij}}{\alpha_{ij} + \alpha_{ji}} = \frac{\gamma_i}{\gamma_i + \gamma_j}.
$
\end{enumerate}
If all weights are strictly positive, then the resulting distribution rule is a \emph{weighted Shapley distribution rule}. 
If $\gamma_i = \gamma_j$ for all $i,j \in V$ the resulting distribution rule is an \emph{unweighted Shapley distribution rule}. 
Note that this case corresponds to an equal-split distribution rule.

\section{Refined Bounds on the Price of Anarchy}\label{sec:poa}
In this section, we first establish our topological bound on the Price of Anarchy for symmetric clustering games and then use it to derive new bounds for some special cases as well as random clustering games.

\subsection{Topological Price of Anarchy Bound}

Our topological bound depends on the \emph{maximum subgraph density} of $G$ which is defined as $\msd{G} = \max_{S \subseteq V} \{|E[S]|/|S|\}$, where $|E[S]|$ is the number of edges in the subgraph induced by $S$. Recall that $\bar\alpha$ refers to the maximum disparity.

\begin{theorem}[Density bound]\label{thm:poa}
Let $\Gamma = (G,c,\alpha,w,q)$ be a symmetric clustering game with $\alpha > 0$.
Then
$
\text{PoA}(\Gamma) \le 1 + \left(1 + \bar{\alpha} \right) \msd{G}
$
and this is tight in general.
\end{theorem}

\begin{proof}
Let $s$ and $s^*$ be a Nash equilibrium and a social optimum, respectively. 
Consider an edge $\{i,j\} \in E$ and assume without loss of generality that $u_i(s) \leq u_j(s)$. 
If $\{i,j\}$ is a coordination edge, then
$
u_i(s) \geq u_i(s_{-i},s_j) \geq \frac{\alpha_{ij}}{\alpha_{ij} + \alpha_{ji}}w_{ij},
$
where $(s_{-i},s_j)$ is the strategy profile in which player $i$ deviates to the color of player $j$ and all other players play according to $s$. 
Suppose $\{i,j\}$ is an anti-coordination edge. If $s_i \neq s_j$, then we trivially have $u_i(s) \geq \alpha_{ij}/(\alpha_{ij} + \alpha_{ji})w_{ij}$ by non-negativity of the weights and individual preferences. If $s_i = s_j$, then the same inequality holds by using the Nash condition for some arbitrary color which is not $s_j$. (We may assume that every player has at least two colors in her strategy set.)
In either case, we conclude that
\begin{equation}\label{eq:weights}
w_{ij} \leq \left(1 + \frac{\alpha_{ji}}{\alpha_{ij}}\right)u_i(s)
\leq \left(1 + \max_{e \in E} \bar{\alpha}_{e}\right) u_i(s)
= \left(1 + \bar{\alpha}\right) u_i(s).
\end{equation}
Moreover, by exploiting that $s$ is a Nash equilibrium and the non-negativity of the edge weights, we obtain for every $i \in V$,
$u_i(s) \geq u_i(s_{-i},s_i^*) \geq q_i(s_i^*).$

Using that the sum of the weights of all satisfied edges in $s^*$ is at most the sum of all edge weights, we obtain 
\begin{eqnarray}
u(s^*)& \leq &\sum_{i \in V} q_i(s_i^*) + \sum_{e = \{i,j\} \in E} w_{ij} 
\le  \sum_{i \in V} u_i(s) +  \left(1 + \bar{\alpha}\right) \sum_{\{i,j\} \in E} \min\{u_i(s),u_j(s)\}. \nonumber 
\end{eqnarray}

If we can find a value $M$ such that
\begin{equation}\label{eq:density}
\sum_{\{i,j\} \in E} \min\{u_i(s),u_j(s)\} \leq M\cdot \sum_{i \in V} u_i(s)
\end{equation}
then it follows that $u(s^*) \leq \left(1 + \left(1 + \bar{\alpha} \right)\cdot M\right) u(s)$.
We show that $M = \max_{S \subseteq V} \{|E[S]|/|S|\}$ satisfies (\ref{eq:density}). 

Let $N(i) = \sset{j \in V}{\set{i, j} \in E}$ be the set of neighbors of $i$. Define
$$
m_i = \big| \sset{ j \in N(i)}{u_i(s) < u_j(s) \text{ or } (u_i(s) = u_j(s) \text{ and } i < j)}\big|
$$
and note that $\sum_{i \in V} m_i = |E|$. We can assume without loss of generality that $\sum_{i \in V} u_i(s) = 1$, since the expression in (\ref{eq:density}) is invariant under multiplication with a constant positive scalar. Moreover, the players may be renamed such that $u_1(s) \leq u_2(s) \leq \dots \leq u_n(s)$. 

We continue by showing that $M$ is an upper bound for the linear program below (in which $u_i = u_i(s)$ and the $m_i$ are considered constants).
$$
\begin{array}{ll@{\quad}l@{\quad}l}
\max  & 
\sum_{i \in V} u_i m_i &
\text{s.t.} & \displaystyle u_1 + u_2 + \dots + u_n = 1, \quad 0 \leq u_1 \leq u_2 \leq \dots \leq u_n
\end{array}
$$
The dual of this program is given by
$$
\begin{array}{ll@{\quad}l@{\quad}l}
\min  &  z &
\text{s.t.} &   - \pi_{i} + \pi_{i+1} + z = m_i, \ \ \ i = 1,\dots,n-1, \quad -\pi_n + z = m_n \\
& & &  \pi_i \geq 0, \ \ \ i = 1,\dots,n, \quad z \in \R \\
\end{array}
$$
We now construct a feasible dual solution. Set
$$
z^* = \max_{l \in V} \left\{ \frac{\sum_{i = l}^{n-1} m_i}{n - l}\right\}.
$$
We will often use that
$
(n - l)z^* \geq \sum_{i = l}^{n-1} m_i
$
for any fixed $l$. 
In particular, with $l = n-1$, we find $z^* \geq m_n$, so that  $\pi_n^* :=  z^* - m_n \geq 0$. Then we define
$
\pi_{n-1}^* := \pi_n^* + z^* - m_{n-1} = 2z^* - (m_{n-1} + m_n) \geq 0.
$
Using induction it then easily follows that
$
\pi_{i}^* := \pi_{i+1}^* + z^* - m_{i} \geq 0
$
for all $i = 1,\dots,n-2$ as well. We have constructed a feasible dual solution with objective function value $z^*$. Using weak duality it follows that for any feasible primal solution $u = (u_1,\dots,u_n)$, we have
$$
\sum_{\{i,j\} \in E} u_i m_i \leq  \max_{l \in V} \left\{ \frac{\sum_{i = l}^{n-1} m_i}{n - l}\right\} \leq \max_{S \subseteq V} \left\{\frac{|E[S]|}{|S|}\right\},
$$
since the term in middle is precisely the density of the induced subgraph on the nodes $l,\dots,n$. This completes the upper bound proof.\medskip

\noindent We continue with showing tightness, already for coordination games.
Let $G = (L \cup R, E)$ be a complete bipartite graph between node-sets $L$ and $R$, with $|L| = l$ and $|R| = r$, and assume that all edges in $E$ are coordination edges.
 We show tightness using a weighted Shapley distribution rule. (That is, for any value of $\max_{\{i,j\} \in E} \bar{\alpha}_{ij}$, there is also some weighted Shapley distribution rule that attains this value.) The nodes in $L$ get a fixed weight $\gamma_l \geq 0$, and the nodes in $R$ get a fixed weight $\gamma_r \geq 0$. 

We define $C = A \cup B \cup \{c_0\}$ where $A$ contains colors $\{a_1,\dots,a_l\}$ and $B = \{b_1,\dots,b_r\}$. We give every player $i \in L$ an individual preference of $\gamma_{l}/(\gamma_l + \gamma_r)$ for colors $a_i$ and $c_0$, every player $j \in R$ an individual preference of $\gamma_{r}/(\gamma_l + \gamma_r)$ for colors $b_j$ and $c_0$, and set all other individual preferences to zero. All edge-weights are set to one. The strategy profile $s$ in which player $i \in L$ plays $a_i$, and $j \in R$ plays $b_j$ is a Nash equilibrium with
$
u(s) = l \cdot \gamma_{l}/(\gamma_l + \gamma_r) + r \cdot \gamma_{r}/(\gamma_l + \gamma_r).
$
The strategy profile $s^*$ in which every player plays color $c_0$ is clearly a social optimum, with cost
$
u(s^*) = l \cdot \gamma_{l}/(\gamma_l + \gamma_r) + r \cdot \gamma_{r}/(\gamma_l + \gamma_r) + r \cdot l
$
It then follows that
$$
\frac{u(s^*)}{u(s)} = 1 + \frac{r \cdot l}{l \cdot \gamma_{l}/(\gamma_l + \gamma_r) + r \cdot \gamma_{r}/(\gamma_l + \gamma_r)}.
$$
By letting $r \rightarrow \infty$, we find a lower bound of $1 + l\cdot(1 + \gamma_l/\gamma_r)$. Note that for $l$ and $r$ fixed, the densest subgraph is the whole graph and has density $lr/(l+r)$ which converges to $l$ as $r \rightarrow \infty$.
\end{proof} \medskip

We use our topological bound to derive deterministic bounds on the Price of Anarchy for two special cases of clustering games. Note that these bounds cannot be deduced from \cite{Anshelevich2014,Feldman2015}.

\begin{corollary}[Planar clustering games]\label{cor:planar}
Let $\Gamma = (G,c,\alpha,w,q)$ be a symmetric clustering game on a planar graph $G$ with $\alpha > 0$.
Then $\text{PoA}(\Gamma) \leq 4 + 3\bar\alpha$. 
\end{corollary}
\begin{proof}
By Euler's formula, $|E(H)|/|V(H)| \leq 3$ for any planar graph $H$. 
Further, any induced subgraph $H$ of a planar graph $G$ is again planar. 
Using this in Theorem~\ref{thm:poa} proves the claim. 
\end{proof}

\begin{corollary}[Equal-split coordination games]\label{cor:poa_unweighted}
Let $G$ be a given undirected graph, and let $\mathcal{G}_G$ be the set of all symmetric coordination games  $\Gamma = (G,c,\vec{1},w,q)$ with equal-split distribution rule on $G$. Then 
$
\text{PoA}(\mathcal{G}_G) = 1 + 2\msd{G}.
$
\end{corollary}
We emphasize that the bound in Corollary \ref{cor:poa_unweighted} is tight on \emph{every} fixed graph topology $G$, rather than only in the \emph{value} of $\rho(G)$.

\begin{proof}[Proof of Corollary \ref{cor:poa_unweighted}]
The upper bound follows directly from Theorem \ref{thm:poa}. 
We prove the lower bound by constructing an instance of a coordination game as follows: Let $S \subseteq V$ be arbitrary and consider the induced subgraph on $S$. Assume without loss of generality that $S = \{1,\dots,\sigma\}$ with $\sigma = |S|$. Define the set of colors as $C = \{c_1,\dots,c_\sigma\} \cup \{c_0\}$. We give every player $i \in S$ an individual preference of one for colors $c_i$ and $c_0$ and zero for all other colors. Further, the individual preferences of all nodes in $V \setminus S$ are set to zero. The weight of all edges in $E[S]$ is set to $2$ and the weight of all edges in $E \setminus E[S]$ is set to zero. 

Consider a strategy profile $s$ in which every player $i \in S$ chooses color $c_i$ and every player $i \notin S$ chooses an arbitrary color. Then $s$ is a Nash equilibrium with social welfare $u(s) = |S|$. On the other hand, the strategy profile $s^*$ in which every player chooses color $c_0$ is a social optimum with social welfare $u(s^*) = |S| + 2|E[S]|$. This implies that $u(s^*)/u(s) = 1 + 2 |E[S]|/|S|$. The result now follows by choosing $S$ as a subset of maximum subgraph density. 
\end{proof}

It is known that the Price of Anarchy of anti-coordination games is $2$ (see, e.g., \cite{Hoefer2007}), which is not reflected by our bound in Theorem \ref{thm:poa}. Intuitively, this suggests that a large Price of Anarchy is caused by the coordination edges of the graph. Theorem~\ref{thm:poa_clustering} reveals that this intuition is correct: it shows 
that the maximum subgraph density with respect to the \emph{coordination edges only} is the determining topological parameter.

\begin{restatable}[Refined density bound]{theorem}{poaclustering}\label{thm:poa_clustering}
Let $\Gamma = (G,c,\vec{1},w,q)$ be a symmetric clustering game with equal-split distribution rule. 
Then
$
\text{PoA}(\Gamma) \leq 5 + 2 \msd{G[E_c]},
$
where $G[E_c]$ is the subgraph induced by the coordination edges $E_c$.  
\end{restatable}
\begin{proof}
The proof is a modification of the proof of Theorem \ref{thm:poa}.
Let $s$ be a Nash equilibrium and $s^*$ a socially optimal strategy profile. For notational convenience, we write $u_i = u_i(s)$ for $i \in V$.  
Moreover, for a strategy profile $t$, we let $E_c(t)$ be the set of all coordination edges satisfied in $t$ and $E_a(t)$ the set of anti-coordination edges that are satisfied in $t$. 

Now, fix some coordination edge $\{a,b\} \in E$, and assume without loss of generality that $u_a \leq u_b$. Then
\begin{equation}\label{eq:clustering1}
u_a(s) \geq u_i(s_{-a},s_b) \geq \frac{1}{2}w_{ab},
\end{equation}
where $(s_{-a},s_b)$ is the strategy profile in which player $a$ deviates to the color of player $b$ and all others play their strategy in $s$. 
Rewriting gives $w_{ab} \leq 2 u_a$.
Moreover, using the non-negativity of the weights $w$ and the definition of a Nash equilibrium, we have for every $i \in V$ that
\begin{equation}\label{eq:clustering2}
u_i(s) \geq u_i(s_{-i},s_i^*) \geq q_i(s_i^*).
\end{equation}
Finally, note that for two arbitrary colors $l_1$ and $l_2$, it follows that
$$
2 u_i(s) \geq u_i(s_{-i},l_1) + u_i(s_{-i},l_2) \geq \frac{1}{2} \sum_{j  : \{i,j\} \in E_a } w_{ij} 
$$
using the Nash condition twice, 
since every anti-coordination edge adjacent to $i$ becomes satisfied for at least one of the two deviations. 
This implies that 
\begin{equation}\label{eq:clustering3}
4 \cdot \sum_{i \in V} u_i(s) \geq \sum_{e \in E_a} w_{ij}.
\end{equation}
Combining (\ref{eq:clustering1}), (\ref{eq:clustering2}) and (\ref{eq:clustering3}), we find
\begin{eqnarray}
u(s^*) & \leq &  \sum_{i \in V} q_i(s_i^*) + \sum_{\{i,j\} \in E_c} w_{ij} + \sum_{\{i,j\} \in E_a } w_{ij}  
\leq  5 \cdot \sum_{i \in V} u_i +  \sum_{\{i,j\} \in E_c} 2 \cdot \min\{u_i,u_j\} \nonumber \\
& \leq & 5 u(s) + 2\cdot \max_{S \subseteq V} \left\{\frac{|E_{c}[S]|}{|S|}\right\} u(s),
\end{eqnarray}
where the final step follows from similar arguments as in the proof of Theorem~\ref{thm:poa}. 
\end{proof}

Using a similar construction as in the proof of Corollary \ref{cor:poa_unweighted} we can also establish a lower bound of $1 + 2\max_{S \subseteq V} \left\{|E_c[S]|/|S|\right\}$.

Note that for anti-coordination games we obtain an upper bound of $5$ which is inferior to the known (tight) bound of $2$. It would be interesting to see whether our topological bound in Theorem \ref{thm:poa_clustering} can be improved to match this bound.

\subsection{Price of Anarchy for Random Coordination Games}
We now turn to our bounds for random coordination games. 
Recall that for random graphs we consider equal-split distribution rules only. 
We first show that for sparse random graphs the Price of Anarchy is constant 
with high probability.

\begin{corollary}[Sparse random coordination games]\label{cor:sparse}
Let $d > 0$ be a constant. 
Let $\mathcal{G}_{G_n}$ be the set of all symmetric coordination games $\Gamma = (G_n,c,\vec{1},w,q)$ on graph $G_n \sim G(n,d/n)$ with equal-split distribution rule.
Then there is a constant $\beta= \beta(d)$ such that $\text{PoA} \left(\mathcal{G}_{G_n}\right) \rightarrow \beta$.
\end{corollary}
\begin{proof}
The maximum subgraph density of a random graph $G_n$ approaches a constant $\beta = \beta(d)$ with high probability \cite{Anantharam2016} (see \cite{Hajek1990} for approximations of this constant).
Combining this with the bound in Corollary \ref{cor:poa_unweighted} proves the claim. 
 \end{proof}

As we show in Theorem~\ref{thm:random_graph}, the result of Corollary \ref{cor:sparse} does not hold for sufficiently dense random graphs if the number of available colors grows large.

\begin{theorem}[Dense random coordination games]\label{thm:random_graph}
Let $0 < d \leq 1$ be a constant and let $(c_n)_{n \in \N} \rightarrow \infty$ be a sequence of available colors. 
Let $\mathcal{G}_{G_n}(c_n)$ be the set of all symmetric coordination games $\Gamma = (G_n,c_n,\vec{1}, w,\mathbf{0})$ on graph $G_n \sim G(n,d)$ with $c_n$ colors, equal-split distribution rule and no individual preferences.
Then there is a constant $\beta = \beta(d)$ such that $\text{PoA} \left(\mathcal{G}_{G_n}(c_n)\right) \geq \beta c_n$.
\end{theorem}

We note that this lower bound holds even for coordination games without individual preferences (as studied in \cite{Feldman2015}). Basically, this bound implies that for dense graph topologies we cannot significantly improve upon the Price of Anarchy bound of $c$ by \cite{Anshelevich2014,Feldman2015}, even if we randomize the graph topology.

\begin{proof}[Proof of Theorem \ref{thm:random_graph}]
We first construct a deterministic instance $\Gamma$ with Price of Anarchy $\Omega(c_n)$ and then show that we can embed this construction into a random graph with high probability. 

Consider a graph $G = (V,E)$ and let $c$ be the number of available colors. 
Let $M = \{e_1,\dots,e_q\} \subseteq E$ be a matching of size at most $c$. Let $V_{M}$ be the set of nodes which are matched in $M$. 
Define the weight of an edge $e \in E$ as $w(e) =  2$ if $e \in M$, $w(e) = 1$ if precisely one of $e$'s endpoints is matched in $M$, and $w(e) = 0$ otherwise. 

Consider the strategy profile $s$ in which the nodes adjacent to $e_i$ play color $i$, for $i = 1,\dots,q$. Note that this is possible because $q \leq c$ by assumption. All other nodes play an arbitrary color; these nodes are irrelevant as all the edges that they are adjacent to have weight zero. In a social optimum $s^*$ all players choose a common color. It follows that $\text{PoA}(\Gamma) \geq |E[V_{M}]|/(2q)$, 
where $|E[V_{M}]|$ is the number of edges in the induced subgraph of $V_{M}$. Note that all these edges have weight at least one.

Now, let $G_n = (V_n,E_n) \sim G(n,d)$ and assume without loss of generality that $V_n = \{1,\dots,n\}$. We claim that with high probability the induced subgraph on nodes $W_n = \{1,\dots,\lceil c_n/4\rceil \}$ contains both $\Omega(c_n^2)$ edges and a perfect matching (if $\lceil c_n/4\rceil$ is odd, we consider the first $\lceil c_n/4\rceil + 1$ nodes). (One may focus on any set of $\lceil c_n/4\rceil$ nodes. The important thing to note is that we need a set of nodes with many edges on its induced subgraph \emph{and} a perfect matching (it is not sufficient to find two different sets each satisfying one of these properties). Moreover, if $c_n \geq 4n$, we consider $W_n = \{1,\dots,n\}$ and then the same argument works.)

The first claim follows from standard arguments. 
Note that 
$$
\mu = \mathbb{E}\{E_n[W_n]\} = d \binom{\lceil c_n/4\rceil}{2} = \Omega(c_n^2).
$$
Using Chernoff's bound, it follows that
$
\mathbb{P}\{E_n[W_n] < \mu/2)\} \leq \exp(-\mu/8) = \exp(-\Omega(c_n^2)/8) \rightarrow 0
$
as $n \rightarrow \infty$ as $(c_n) \rightarrow \infty$.	 
The second claim relies on the following result (see, e.g., \cite{Frieze2015}): 
For every fixed $0 < d \leq 1$ it holds that
$
\lim_{n \rightarrow \infty} \mathbb{P}_{G_n \sim G(n,d)}\{G_n \text{ contains a perfect matching}\} = 1.
$
By applying this result to the induced subgraph on $W_n$ and using that $c_n$ approaches infinity as $n \rightarrow \infty$, the claim follows. (Note that here we implicitly use that the intersection of two probabilistic events which occur with high probability also occurs with high probability.)

Combining this with the deterministic bound on the Price of Anarchy derived above concludes the proof. 
\end{proof}

\section{Convergence of Best-Response Dynamics}\label{sec:existence}

In this section, we derive our characterization results for the convergence of best-response dynamics in symmetric clustering games and for the existence of pure Nash equilibria in symmetric coordination games. 
(Recall that best-response dynamics are said to \emph{converge} if any sequence of player deviations, where in each step the deviating player chooses a most profitable deviation, converges in a finite number of steps to a (pure) Nash equilibrium.)
Basically, for symmetric clustering games our characterization shows that best-response dynamics are guaranteed to converge to a pure Nash equilibrium if and only if $\alpha$ is a generalized weighted Shapley distribution rule.
For the special case of symmetric coordination games with $c \geq 3$, we can further strengthen this characterization result and show that a pure Nash equilibrium is guaranteed to exist if and only if $\alpha$ is a generalized weighted Shapley distribution rule. This complements a results of Anshelevich and Sekar \cite{Anshelevich2014}.

\subsection{Symmetric Clustering Games}
We provide a characterization of distribution rules that guarantee the convergence of best-response dynamics in symmetric clustering games. 

\begin{theorem}[Best-response convergence]\label{thm:best_response_2}
Let $\mathcal{G}_{G,c,\alpha}$ be the set of all symmetric clustering games $\Gamma = (G,c,\alpha, w, q)$ on a fixed graph $G$ with $c$ common colors and distribution rule $\alpha$.
Then best-response dynamics are guaranteed to converge to a pure Nash equilibrium for every clustering game in $\mathcal{G}_{G,c,\alpha}$ if and only if $\alpha$ corresponds to a generalized weighted Shapley distribution rule.
\end{theorem}

In general, this characterization does not hold if the condition of ``guaranteed convergence of best-response dynamics" is replaced by ``guaranteed existence of a pure Nash equilibrium" (as in \cite{Gopalakrishnan2013} or \cite{Chen2010}): 
There are settings where on a fixed graph $G$, a pure Nash equilibrium is guaranteed to exist even if $\alpha$ is not a generalized weighted Shapley distribution rule, e.g., in the case of $c = 2$, or in coordination games with no individual preferences. 

The proof of Theorem \ref{thm:best_response_2} relies on the following lemma. In the proofs of Lemma \ref{lem:best_response} and Theorem \ref{thm:best_response_2}, player or edges indices are always modulo $n$. 

\begin{lemma}\label{lem:best_response}
Consider a symmetric clustering game $(H,2,\alpha,w,q)$ on a cycle $H =  \langle 1, \dots, n \rangle$  with $n$ players and $c = 2$ colors.
If for every strategy profile $s$ it is a best-response for every player $i$ to choose a color that satisfies at least edge $w_{\{i,i+1\}}$, then there exists 
a best-response sequence that does not converge to a Nash equilibrium.
\end{lemma}

\begin{proof}
We first construct an initial state $s^0$ using only colors $k_1$ and $k_2$. Set $s_{1}^0 = k_1$, and iteratively, for $i = 2,\dots,n-1$, set $s_{i}^0$ such that edge $w_{\{i,i+1\}}$ is satisfied (using only colors $k_1$ and $k_2$). The color for $s_n$ is chosen in such a way that at least one of the edges $w_{n-1,n}$ or $w_{n,1}$ is not satisfied (this can always be done, since if color $k_1$ would satisfy both edges, that color $k_2$ would satisfy neither, and vice versa). Now, either $i)$ precisely $n-1$ edges of the cycle are satisfied in $s^0$, or $ii)$ precisely $n-2$ edges are satisfied in $s^0$ (and two consecutive edges are not). Both situations are illustrated below. 

\textbf{Case i):} There are precisely $n-1$ edges satisfied. Note that currently edge $\{n,1\}$ is not satisfied. Therefore, by assumption, it is a best-response for player $n$ to switch its other color. But the situation after this switch is isomorphic to the starting profile $s^0$, that is, if we would have started the numbering at node $n$ instead of node $1$. Therefore, we can repeat the same argument, and in particular, after $2n$ of such best-response steps (in which, roughly speaking, the unsatisfied edge moves over the cycle), we are back in $s^0$.

\begin{figure}[h!]
\centering
\begin{tikzpicture}[scale=3]
\coordinate (A1) at (-0.25,0); 
\coordinate (A2) at (0.25,0);
\coordinate (M1) at (-0.5,0.375);
\coordinate (M2) at (0.5,0.375);
\coordinate (T1) at (-0.25,0.75);
\coordinate (T2) at (0.25,0.75);

\node at (A1) [circle,scale=0.7,fill=black] {};
\node (a1) [below=0.1cm of A1]  {$k_1$};
\node at (A2) [circle,scale=0.7,fill=black] {};
\node (a2) [below=0.1cm of A2]  {$k_1$};
\node at (M1) [circle,scale=0.7,fill=black] {};
\node (m1) [left=0.1cm of M1]  {$k_2$};
\node at (M2) [circle,scale=0.7,fill=black] {};
\node (m2) [right=0.1cm of M2]  {$k_2$};
\node at (T1) [circle,scale=0.7,fill=black] {};
\node (t1) [above=0.1cm of T1]  {$k_1$};
\node (t1b) [below=0.1cm of T1]  {$1$};
\node at (T2) [circle,scale=0.7,fill=black] {};
\node (t2) [above=0.1cm of T2]  {$k_1$};

\path[every node/.style={sloped,anchor=south,auto=false}]
(T1) edge[-,very thick] node {$+$} (T2) 
(T2) edge[-,very thick] node {$-$} (M2)
(M2) edge[-,very thick] node {$-$} (A2) 
(A2) edge[-,very thick] node {$+$} (A1) 
(A1) edge[-,very thick] node {$-$} (M1) 
(M1) edge[-,dashed] node {$+$} (T1);  
\end{tikzpicture}
\quad
\begin{tikzpicture}[scale=3]
\coordinate (A1) at (-0.25,0); 
\coordinate (A2) at (0.25,0);
\coordinate (M1) at (-0.5,0.375);
\coordinate (M2) at (0.5,0.375);
\coordinate (T1) at (-0.25,0.75);
\coordinate (T2) at (0.25,0.75);


\node at (A1) [circle,scale=0.7,fill=black] {};
\node (a1) [below=0.1cm of A1]  {$k_1$};
\node at (A2) [circle,scale=0.7,fill=black] {};
\node (a2) [below=0.1cm of A2]  {$k_1$};
\node at (M1) [circle,scale=0.7,fill=black] {};
\node (m1) [left=0.1cm of M1]  {$k_1$};
\node at (M2) [circle,scale=0.7,fill=black] {};
\node (m2) [right=0.1cm of M2]  {$k_2$};
\node at (T1) [circle,scale=0.7,fill=black] {};
\node (t1) [above=0.1cm of T1]  {$k_1$};
\node (t1b) [below=0.1cm of T1]  {$1$};
\node at (T2) [circle,scale=0.7,fill=black] {};
\node (t2) [above=0.1cm of T2]  {$k_1$};

\path[every node/.style={sloped,anchor=south,auto=false}]
(T1) edge[-,very thick] node {$+$} (T2) 
(T2) edge[-,very thick] node {$-$} (M2)
(M2) edge[-,very thick] node {$-$} (A2) 
(A2) edge[-,very thick] node {$+$} (A1) 
(A1) edge[-,dashed] node {$-$} (M1) 
(M1) edge[-,very thick] node {$+$} (T1);  
\end{tikzpicture}

\caption{The cycles are numbered clockwise starting at $1$ (indicated at top left). The satisfied edges are bold, and the unsatisfied edges are dashed. On the left the starting state $s^0$ and on the right the profile after one best-response move as described above.}
\label{fig:case1_after}
\end{figure}
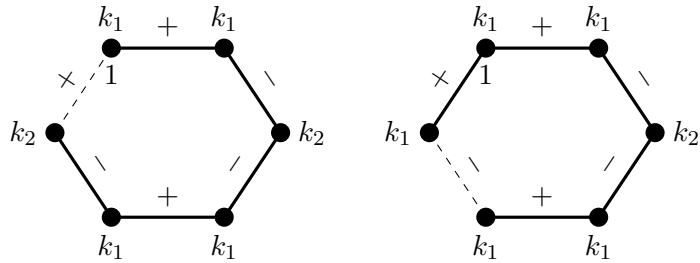 
\textbf{Case ii):} There are precisely $n-2$ edges satisfied except for the two consecutive edges $\{n_1,n\}$ and $\{n,1\}$. If player $n$ would switch to its other color (which is a best-response move) then we would find a Nash equilibrium, however, we do not choose player $n$. Instead we let player $n-1$ switch to its other color (which is a best-response move by assumption), then afterwards, it is a best-response for player $n-2$ to switch as well (in order to satisfy edge $\{n-2,n-1\}$, and we continue this in decreasing player order up until (and including) player $2$. In particular, we are then in the situation were again precisely $n-2$ edges are satisfied except two consecutive edges, which are now $\{n,1\}$ and $\{1,2\}$. This situation is equivalent to the starting state $s^0$, and in particular by repeating this process $n$ times, we are back in $s^0$. 
This completes the proof.
 \end{proof}
\begin{figure}[h!]
\centering
\begin{tikzpicture}[scale=3]
\coordinate (A1) at (-0.25,0); 
\coordinate (A2) at (0.25,0);
\coordinate (M1) at (-0.5,0.375);
\coordinate (M2) at (0.5,0.375);
\coordinate (T1) at (-0.25,0.75);
\coordinate (T2) at (0.25,0.75);

\node at (A1) [circle,scale=0.7,fill=black] {};
\node (a1) [below=0.1cm of A1]  {$k_1$};
\node at (A2) [circle,scale=0.7,fill=black] {};
\node (a2) [below=0.1cm of A2]  {$k_1$};
\node at (M1) [circle,scale=0.7,fill=black] {};
\node (m1) [left=0.1cm of M1]  {$k_1$};
\node at (M2) [circle,scale=0.7,fill=black] {};
\node (m2) [right=0.1cm of M2]  {$k_2$};
\node at (T1) [circle,scale=0.7,fill=black] {};
\node (t1) [above=0.1cm of T1]  {$k_1$};
\node (t1b) [below=0.1cm of T1]  {$1$};
\node at (T2) [circle,scale=0.7,fill=black] {};
\node (t2) [above=0.1cm of T2]  {$k_1$};

\path[every node/.style={sloped,anchor=south,auto=false}]
(T1) edge[-,very thick] node {$+$} (T2) 
(T2) edge[-,very thick] node {$-$} (M2)
(M2) edge[-,very thick] node {$-$} (A2) 
(A2) edge[-,very thick] node {$+$} (A1) 
(A1) edge[-,dashed] node {$-$} (M1) 
(M1) edge[-,dashed] node {$-$} (T1);  
\end{tikzpicture}
\quad
\begin{tikzpicture}[scale=3]
\coordinate (A1) at (-0.25,0); 
\coordinate (A2) at (0.25,0);
\coordinate (M1) at (-0.5,0.375);
\coordinate (M2) at (0.5,0.375);
\coordinate (T1) at (-0.25,0.75);
\coordinate (T2) at (0.25,0.75);

\node at (A1) [circle,scale=0.7,fill=black] {};
\node (a1) [below=0.1cm of A1]  {$k_2$};
\node at (A2) [circle,scale=0.7,fill=black] {};
\node (a2) [below=0.1cm of A2]  {$k_2$};
\node at (M1) [circle,scale=0.7,fill=black] {};
\node (m1) [left=0.1cm of M1]  {$k_1$};
\node at (M2) [circle,scale=0.7,fill=black] {};
\node (m2) [right=0.1cm of M2]  {$k_1$};
\node at (T1) [circle,scale=0.7,fill=black] {};
\node (t1) [above=0.1cm of T1]  {$k_1$};
\node (t1b) [below=0.1cm of T1]  {$1$};
\node at (T2) [circle,scale=0.7,fill=black] {};
\node (t2) [above=0.1cm of T2]  {$k_2$};

\path[every node/.style={sloped,anchor=south,auto=false}]
(T1) edge[-,dashed] node {$+$} (T2) 
(T2) edge[-,very thick] node {$-$} (M2)
(M2) edge[-,very thick] node {$-$} (A2) 
(A2) edge[-,very thick] node {$+$} (A1) 
(A1) edge[-,very thick] node {$-$} (M1) 
(M1) edge[-,dashed] node {$-$} (T1);  
\end{tikzpicture}
\caption{The cycle is numbered clockwise starting at $1$ (indicated at top left). The satisfied edges are bold, and the unsatisfied edges are dashed. On the left the starting state $s^0$, and on the right the profile after the first $n-2$ best-response moves as described above.}
\label{fig:case2}
\end{figure}
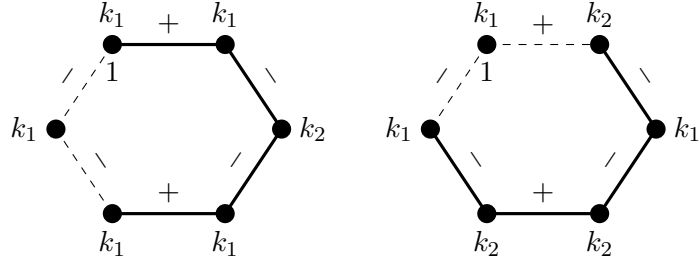 

\begin{proof}[Proof of Theorem \ref{thm:best_response_2}]
If $\alpha$ corresponds to a generalized weighted Shapley distribution rule, then the convergence of best-response dynamics follows immediately from the fact that the game can be modeled as a resource allocation game (see Appendix \ref{sec:app_resource}). Such games, with generalized weighted Shapley distribution rules, are potential games, see, e.g., \cite{Gopalakrishnan2013} for details. We now continue with the other direction, i.e., assume that best-response dynamics always converge.

We first look at the subgraph of $G$ consisting of edges $e = \{i,j\}$ for which $\alpha_{ij}, \alpha_{ji} > 0$. Let $Q_1,\dots,Q_r$ be the node sets of the connected components in this subgraph (isolated nodes are also connected components). We then define the digraph $D = ([r], A)$ where there is a directed arc from $a$ to $b$ with $a,b \in [r]$ if and only if there is an edge $e = \{i,j\} \in E$ such that $\alpha_{ij} = 0$ (and thus $\alpha_{ji} > 0$) with $i \in Q_a$ and $j \in Q_b$.

We claim that $D$ must be acyclic, i.e., it does not contain cycles or self-loops. For both cases we can construct a counter-example (see below). Any topological ordering for the graph $D$ then gives rise to a permutation $\sigma$ of the nodes in $V$ satisfying condition $i)$ in the definition of a generalized weighted Shapley scheme.

We first give a counter-example for the case that there is a self-loop. If a component $Q$ contains a self-loop, then there is some cycle $H$ in $G$, numbered $v_1$ through $v_h$ for some orientation, with the property that the edge $e_h = \{v_{h},v_1\}$ has $\alpha_{v_hv_1} > 0$ and $\alpha_{v_1v_h} = 0$ and for all other edges $e_i = \{v_i,v_{i+1}\}$ both $\alpha_{v_iv_{i+1}}, \alpha_{v_{i+1}v_i} > 0$. 
Now, fix some weight $w_{v_1v_2} > 0$ and iteratively define the weights $w_{v_iv_{i+1}}$ so that 
$$
\frac{\alpha_{v_iv_{i-1}}}{\alpha_{v_iv_{i-1}} + \alpha_{v_{i-1}v_i}} w_{v_{i-1}v_i} <\frac{\alpha_{v_iv_{i+1}}}{\alpha_{v_iv_{i+1}} + \alpha_{v_{i+1}v_i}} w_{v_iv_{i+1}}.
$$
All other edge-weights (of edges not in $H$) are also set to zero. Individual preferences (for every player) are set to $K = \sum_{e \in E} w_e$ for two fixed colors $k_1$ and $k_2$, and zero otherwise. 
In particular this means that for all players $v_i$, for $i = 2,\dots,h$, a best-response move will always be either color $k_1$ or $k_2$ in any strategy profile. (This is not necessary in the case of coordination games, because if every player starts with either $k_1$ or $k_2$, then any best-response move will also be one of these two colors, since deviating to some third color would give a pay-off of zero, so in that case we can set all individual preferences to zero. For the case where there are also anti-coordination edges, this does not work, since then deviating to some third color might the (only) best-response if a player is adjacent to two anti-coordination edges. Nevertheless, deviating to $k_1$ or $k_2$ will still be a better-response (but not a best-response). Of course, if $c = 2$, then there is no problem. These arguments justify the statements in  Remark \ref{rem:best_response_2}.) Moreover, note that the utility of player $v_h$ derived from edge $e_h$ is zero (which is strictly lower that its utility derived from edge $e_1$). Now we can apply Lemma \ref{lem:best_response} to conclude that best-response dynamics are then not guaranteed to converge. 

For a cycle in $D$ the argument is similar. Then the cycle $H$ contains consecutive sections consisting of arcs with both shares strictly positive and an arc with one share zero at the end. On every such section we can choose weights $w$ as above and then use Lemma \ref{lem:best_response} again.

For a fixed connected component $Q = Q_d$ for some $d$, with $|V(Q)| = q$, we claim that there exist weights $\gamma_Q = (\gamma_1,\dots,\gamma_q)$ so that 
\begin{equation}\label{eq:gamma_share}
\frac{\alpha_{ij}}{\alpha_{ij} + \alpha_{ji}} = \frac{\gamma_i}{\gamma_i + \gamma_j}
\end{equation}
for all $i,j \in \{1,\dots,q\}$. It is clear that if we can find such a vector $\gamma$, and multiply every $\gamma_i$ with a fixed constant $d > 0$, then the weights $d \cdot \gamma_i$ also satisfy the desired condition. In particular, this implies that we can always fix $\gamma_1$ as we like without loss of generality. For every player $j$ adjacent to player $1$, the equation (\ref{eq:gamma_share}) then uniquely determines $\gamma_j$, that is, we have
$
\gamma_j = \frac{\alpha_{j1}}{\alpha_{1j}} \gamma_1.
$
Moreover, by repeating this argument we can construct a spanning tree $T$ (on the nodes in $Q$) with the property that
$$
\gamma_j = \frac{\alpha_{p_2p_1}\alpha_{p_3p_2}\dots\alpha_{p_{h_j}p_{h_j-1}}}{\alpha_{p_1p_2}\alpha_{p_2p_3}\dots\alpha_{p_{h_j-1}p_{h_j}}}\gamma_1
$$
where $(p_1,p_2,\dots,p_{h_j})$ is the unique path from player $1$ to player $j$ in $T$. Now, if there is some edge $e \in E(G) \setminus E(T)$ with the property that (\ref{lem:best_response}) is not satisfied, then it follows that
$$
\alpha(H) :=\frac{\alpha_{p_2p_1}\alpha_{p_3p_2}\dots\alpha_{p_{1}p_{h_j}}}{\alpha_{p_1p_2}\alpha_{p_2p_3}\dots\alpha_{p_{h_j}p_{1}}} \neq 1
$$
where $H = (p_1,\dots,p_h,p_1)$, with $p_1 = 1$, is the unique cycle in $Q$ containing edge $e$. (The expression $\alpha(H) = 1$ is analogue to the cyclic consistency property of Gopalakrishnan et al. \cite{Gopalakrishnan2013}.) We let $e_i = \{p_i,p_{i+1}\}$ for $i = 1,\dots,h$ modulo $h$.
Assume without loss of generality that $\alpha(H) < 1$. Then there exists a constant $\epsilon > 0$ so that 
$
(1+\epsilon)^n\alpha(H) < 1.
$ 
Fix some weight $w_{e_1} > 0$ and  iteratively define the weights $w_{e_i}$ so that 
\begin{equation}\label{eq:w}
\frac{\alpha_{p_ip_{i-1}}}{\alpha_{p_ip_{i-1}} + \alpha_{p_{i-1}p_{i}}} w_{e_{i-1}} = \frac{\alpha_{p_ip_{i+1}}}{\alpha_{p_ip_{i+1}} + \alpha_{p_{i+1}p_i}} w_{e_i}
\end{equation}
for $i = 2,\dots,h$. We then define the weights $w'_{p_ip_{i+1}} = (1+\epsilon)^i\cdot w_{p_ip_{i+1}}$ for $i = 1,\dots,h$. All other edge-weights (of edges not in $H$) are also set to zero. Individual preferences (for every player) are set to $K = \sum_{e \in E} w_e$ for two fixed colors $k_1$ and $k_2$, and zero otherwise (similar as in the first part of the proof). Then it follows directly that for all players $i = 2,\dots,h$ it is always a best-response to choose a color ($k_1$ or $k_2$) that satisfies edge $w_{i,i+1}'$. Moreover, this is also true for player $1$, which can be seen as follows. Suppose it is not true, then
$$
\frac{\alpha_{p_1p_h}}{\alpha_{p_1p_h} + \alpha_{p_hp_1}} w_{e_h} \cdot (1+\epsilon)^n \geq \frac{\alpha_{p_1p_2}}{\alpha_{p_1p_2} + \alpha_{p_{2}p_1}} w_{e_1}.
$$
If we multiply all equalities in (\ref{eq:w}) with this inequality, then after simplification, we find $(1+\epsilon)^n\alpha(H) \geq 1$, which contradicts with the choice of $\epsilon$. We can now again apply Lemma \ref{lem:best_response} and that concludes the proof. 
\end{proof}

\begin{remark}\label{rem:best_response_2} 
Theorem \ref{thm:best_response_2} remains valid also for various settings without individual preferences. For example, this holds for coordination games (corresponding to certain models in \cite{Anshelevich2014,Feldman2015}) and for  general clustering games with $c = 2$. (In general, this is not true if $c \geq  3$. E.g., consider a cycle of length three with only anti-coordination edges.) 
\end{remark}

\subsection{Symmetric Coordination Games} 
We next consider the special case of symmetric coordination games in which the common strategy set contains $c \ge 3$ colors. We can strengthen the characterization result of Theorem \ref{thm:best_response_2} in this case. More specifically, we prove in Theorem \ref{thm:pure_nash} that a pure Nash equilibrium is guaranteed to exist if and only if $\alpha$ is a generalized weighted Shapley distribution rule. 
This complements a result of Anshelevich and Sekar \cite{Anshelevich2014}.

\begin{theorem}\label{thm:pure_nash}
Let $
\mathcal{G}_{G,c,\alpha} = \{ (G,c,\alpha, w, q) : w : E \rightarrow \R_{\geq 0},\ q_i: [c] \rightarrow \R_{\geq 0},\ i \in V \}
$ be the set of all symmetric coordination games
on graph $G$ with common strategy set $\{1,\dots,c\}$ for $c \geq 3$ and distribution rule $\alpha$. Then a pure Nash equilibrium is guaranteed to exist for every game in $\mathcal{G}_{G,c,\alpha}$ if and only if $\alpha$ corresponds to a generalized weighted Shapley distribution rule.
\end{theorem}

Our arguments are conceptually similar to those of Gopolakrishnan et al. \cite{Gopalakrishnan2013}, however, they are technically different. We essentially show a similar result as in \cite{Gopalakrishnan2013}, but for a more restricted setting than the resource allocation games considered there. 
We elaborate on the connection between Theorem \ref{thm:pure_nash} and the work in \cite{Gopalakrishnan2013} in Appendix \ref{sec:app_resource}.
Nevertheless, the result in Theorem \ref{thm:pure_nash} allows us to fully characterize which distribution $\alpha$ guarantee equilibrium existence, thereby completing results of Anshelevich and Sekar \cite{Anshelevich2014}, who only partially address this question.

In particular, Anshelevich and Sekar \cite{Anshelevich2014} provide an example showing that for general distribution rules, pure Nash equilibria are not guaranteed to exist. On the positive side, they show that if the distribution rule has the so-called \emph{correlated coordination condition}, then pure Nash equilibria are guaranteed to exist. This condition is actually the same as saying that the local distribution rule corresponds to a weighted Shapley distribution rule, and the proof of Theorem 1 in \cite{Anshelevich2014} is essentially a direct consequence of the work of Hart and Mas-Collel \cite{Hart1989} who characterize the (weighted) Shapley value in terms of a (weighted) potential function. Theorem \ref{thm:pure_nash} allows us to precisely characterize which distribution rules $\alpha$ guarantee (pure) equilibrium existence in symmetric coordination games, for arbitrary weight functions and individual preferences. In particular, we note that generalized weighted Shapley distribution rules (see preliminaries) still guarantee equilibrium existence, following from \cite{Gopalakrishnan2013}, by observing that these coordination games are resource allocation games (see Appendix \ref{sec:app_resource}), and show that these distribution rules are also necessary in a certain sense, already in the case of three colors, which is the best possible. 

\begin{proof}[Proof of Theorem~\ref{thm:pure_nash}]
The proof is of a similar nature as that of Theorem \ref{thm:best_response_2}. We follow the notation as introduced there. Consider the graph $D$ and suppose it has a self-loop, then we can construct weights for the edges of the cycle $H$ as in the proof of Theorem \ref{thm:best_response_2}. In addition, we now also define individual preferences (which were set to zero in the proof of Theorem \ref{thm:best_response_2}). Let $C = \{1,\dots,c\}$ and let $M > 0$ be some big constant. We define
$$
q_{v_{h-1}}(j) = \left\{\begin{array}{ll}
M + \delta & \text{ if } j = 1 \\
M & \text{ if } j = 2 \\
0 & \text{ else, }
\end{array}\right. 
\qquad
q_{v_h}(j) = \left\{\begin{array}{ll}
M + \delta & \text{ if } j = 2 \\
M & \text{ if } j = 3 \\
0 & \text{ else, }
\end{array}\right.
$$
and
$$ 
q_{v_i}(j) = \left\{\begin{array}{ll}
M + \delta & \text{ if } j = 3 \\
M & \text{ if } j = 1 \\
0 & \text{ else, }
\end{array}\right. 
$$
for $i = 1,\dots,h-2$ where $\delta > 0$ is chosen sufficiently small such that
$$
\frac{\alpha_{v_i, e_{i-1}}}{\alpha_{v_i, e_{i-1}} + \alpha_{v_{i-1}, e_{i-1}}} w_{v_{i-1}v_i} + \delta < \frac{\alpha_{v_i, e_{i}}}{\alpha_{v_i, e_{i-1}} + \alpha_{v_{i}, e_{i}}} w_{v_iv_{i+1}}
$$
is true for all $i = 1,\dots,h$. Through inspection it can be seen that this instance does not have a pure Nash equilibrium. Similar arguments can be carried out inside a fixed component $Q$ (as in the second part of the proof of Theorem \ref{thm:best_response_2}). 
\end{proof}

\section{Results for Asymmetric Coordination Games}\label{sec:ext}
In this section, we present our results for asymmetric coordination 
games. We focus on coordination games with equal-split distribution rule and no individual preferences. 

\subsection{Approximate Nash Equilibria}
Apt et al. \cite{Apt2015} show that the $(1,1)$-PoA of coordination games is unbounded if $c \ge n+1$. Notably, this holds for arbitrary graph topologies with unit weights and without individual preferences. 
We slightly generalize this observation. We show that the Price of Anarchy is unbounded if and only if $c \ge \chi(G)+1$, where $\chi(G)$ is the chromatic number of $G$. 

\begin{theorem}\label{lem:1poa}
Let $\mathcal{G}_{G}(c)$ be the set of all coordination games $\Gamma = \{G,c,(S_i)_{i \in V},\mathbf{1},\mathbf{1},\mathbf{0}\}$ on graph $G$ with $c$ colors and equal-split distribution rule. Then $(\epsilon,1)\text{-PoA}(\mathcal{G}_{G}(c)) = \infty$ if $c \geq \chi(G) + 1$ and finite if $c < \chi(G)$.
\end{theorem}
\begin{proof}[Proof of Theorem \ref{lem:1poa}]
For a given coloring of $G$ with colors $a_1,\dots,a_{\chi(G)}$, we assign a strategy set of $\{a_i,a_0\}$ to all nodes that are colored with color $i \in \{1,\dots,\chi(G)\}$. In particular, the strategy profile $s$ in which every player chooses its color $a_i$ from the coloring is a pure Nash equilibrium with utility $u(s) = 0$, whereas the profile $s^*$ in which every players chooses $a_0$ is a socially optimal profile with $u(s^*) = |E(G)|$.  This shows unboundedness if $c \geq \chi(G) + 1$.

If $c < \chi(G)$ then in every strategy profile $s$ there is at least one edge $e \in E(G)$ such that its endpoints have the same color in $s$. This show boundedness.
\end{proof}

We can exploit the above insight to prove that if the number of colors $c$ is a constant then the Price of Anarchy is unbounded for sparse random graphs, while it is bounded by some constant for dense random graphs:

\begin{theorem}[Constant strategy sets]\label{thm:constant}
Let $\epsilon \geq 1$ and let $c^* \geq 3$ be a given integer.  
Let $\mathcal{G}_{G_n,c^*}$ be the set of all coordination games $\Gamma = (G_n,c^*,(S_i)_{i\in V},\mathbf{1},\mathbf{1},\mathbf{0})$ on graph $G_n \sim G(n,p)$ with strategy sets $S_i \subseteq [c^*]$ for every player $i$. Then there exists a constant $d = d(c^*)$ such that for $p \le d/n$, we have
$
\lim_{n \rightarrow \infty} \mathbb{P}_{G_n \sim G(n,p)}\left\{ (\epsilon,1)\text{-PoA}\left(\mathcal{G}_{G_n,c^*}\right) = \infty \right\} = 1.
$
In contrast, if $p \in (0,1)$ is constant, then there exists a constant $\beta_0 = \beta_0(p,c^*)$ such that
$
\lim_{n \rightarrow \infty} \mathbb{P}_{G_n \sim G(n,p)}\left\{ (\epsilon,1)\text{-PoA}\left(\mathcal{G}_{G_n,c^*}\right) \leq \beta_0 \right\} = 1.
$
\end{theorem}
\begin{proof}
For sparse Erd\H{o}s-R\'enyi graphs $G \sim G(n,d/n)$ for constant $d > 0$, it is known that the chromatic number is a constant $\beta = \beta(d)$ with high probability, see \cite{Luczak1991}. It is not hard to see that this constant is non-decreasing in $d$. On the other hand, for $G \sim G(n,p)$ with $p$ constant independent of $n$, a standard argument shows that with high probability as $n \rightarrow \infty$, there is a constant $\beta = \beta(p)$ such that every subset of more than $\beta n$ nodes contains $\Theta(n^2)$ edges. Since the maximum number of colors in a strategy set is bounded by $c^*$, there is at least one subset of $n/c^*$ nodes such that all players in this set play the same color in a given Nash equilibrium $s$. This means that the social cost of any Nash equilibrium is $\Theta(n^2)$. Moreover, since all edge-weights are one, it follows that the social cost of an optimal strategy profile is $O(n^2)$. This proves that there is a constant $\beta_0$ dependent on both $c^*$ and $p$.
\end{proof}

\subsection{Approximate $k$-Strong Equilibria}
In general, approximate Nash equilibria are not guaranteed to exist in asymmetric coordination games (see, e.g., \cite{Apt2015}). 
In this section, we therefore consider the Price of Anarchy of $(\epsilon, k)$-equilibria with $k \ge 2$. 
It is known that the $(\epsilon,k)$-PoA of coordination games is between $2\epsilon(n-1)/(k-1) + 1 - 2 \epsilon$ and $2\epsilon(n-1)/(k-1)$ for $k \geq 2$ \cite{Rahn2015}. In particular, the $(\epsilon,k)$-PoA grows like $\Theta(\epsilon n)$ if $k$ is a constant. 

We derive a topological bound on the $(\epsilon, k)$-Price of Anarchy which depends on the maximum degree $\Delta(G)$ of the graph $G$.

\begin{theorem}[Degree bound]\label{lem:max_degree}
Let $\epsilon \geq  1$, $k \geq 2$, $c \geq 3$, and let $G$ be an arbitrary graph. Let $\mathcal{G}_{G}(c)$ be the set of all coordination games $\Gamma = (G,c,(S_i),\vec{1}, w,\mathbf{0})$ on graph $G$ with $c$ colors, equal-split distribution rule and no individual preferences.
Then 
$$
\epsilon \cdot \max\bigg\{1,\frac{\Delta(G)}{k-1} - 1\bigg\} \leq (\epsilon,k)\text{-PoA}(\mathcal{G}_{G}(c)) \leq 2\epsilon\cdot\Delta(G).
$$
\end{theorem}

\begin{proof}
We first construct the lower bound. It is sufficient to take $\epsilon = 1$. We may assume that $\Delta(G) > k - 1$ (otherwise we have a trivial lower bound of one). We consider a game with three colors $\{a,b,c\}$. Let $i \in V$ be a node of maximum degree, and let $j_1^*,\dots,j_{k-1}^* \in N(i)$ be $k-1$ fixed neighbors of $i$. We give players $i$ and $j^*_l$ for $l = 1,\dots,k-1$ a strategy set of $\{a,b\}$, and all other nodes a strategy set of $\{a,c\}$. Moreover, edges $\{i,j\}$ get a weight of $w_{ij} = 1$ for $j = j_1^*,\dots,j_{k-1}^*$, $w_{ij} = \epsilon$ for all $j \in N(i) \setminus \{j_1^*,\dots,j_{k-1}^*\}$ and all other edges a weight of zero. It is not hard to see that the strategy profile $s$ in which $s_v = b$ for $v = i,j_1^*,\dots,j_{k-1}^*$, and $s_v = c$ otherwise, is a $k$-equilibrium with utility $u(s) = k-1$. The strategy profile $s^*$ in which all players choose color $a$ is clearly a socially optimal state with utility $u(s^*) = \epsilon(\Delta(G) - k + 1) + k$. This proves the lower bound.

It remains to proof the upper bound. Consider an instance $\Gamma = (G,c,(S_i)_{i\in V},\mathbf{1}, w, \mathbf{0})$. Let $s$ be an $(\epsilon,k)$-equilibrium and let $s^*$ be an optimal strategy profile. Let $V = S \cup T$ be a partition of the node set, where $S = \{i \in V : u_i(s) > 0\}$ and $T = \{i \in V : u_i(s) = 0\}$. 

Let $i,j \in T$ and suppose that $e = \{i,j\} \in E$. We claim that either $w_e = 0$, or $e$ is unsatisfied in $s^*$. Suppose that $w_e > 0$ and $e$ is satisfied in $s^*$. Then, in particular, it follows that $i$ and $j$ both have a color $c'$ in their strategy set, i.e., $S_i \cap S_j \neq \emptyset$. Since $u_i(s) = u_j(s) = 0$, this means that they can (jointly) profitably deviate to $c'$, contradicting the fact that $s$ is a $k$-equilibrium. That is, either one the players chose $c'$ in $s$, in which case the other player can deviate to $c'$ to improve her utility, or $i$ and $j$ can jointly deviate to $c'$ which is feasible because $k \geq 2$. 

The above implies that
$$
u(s^*) \leq \sum_{\{i,j\} \in E(s^*) : \{i,j\} \cap S \neq \emptyset} w_{ij} = \sum_{\{i,j\} \in E(s^*) : \{i,j\} \cap S \neq \emptyset \text{ and } w_{ij} > 0} w_{ij},
$$
where $E(s^*)$ is the set of satisfied edges in $s^*$. 
We now show that the latter summation is at most 
$
2 \epsilon \Delta(G)\cdot u(s),
$
which completes the proof.

First, let $i \in S$ and $j \in T$, and suppose that $e = \{i,j\}$ is satisfied in $s^*$ with $w_e > 0$. The fact that $e$ is satisfied in $s^*$ implies that $i$ and $j$ have a common color $c'$ in their strategy sets. By definition, we have $u_j(s) = 0$, so it must be that $\epsilon \cdot u_i(s) \geq  w_{ij}/2$ otherwise $i$ and $j$ could (jointly) profitably deviate to $c'$.
Secondly, let $i \in S$ and $j \in S$, and suppose that $e = \{i,j\}$ is satisfied in $s^*$ with $w_e > 0$. Similar arguments imply that either $\epsilon \cdot u_i(s) \geq w_{ij}/2$ or $\epsilon \cdot u_j(s) \geq  w_{ij}/2$ (or both). 

In particular, this implies that the edges in 
$
\sset{e \in E(s^*)}{w_e > 0 \text{ and } e \cap S \neq \emptyset}
$ 
can be partitioned into sets $E_1,\dots,E_{|S|}$ defined as
$
E_i = \sset{\{i,j\}}{i \prec j \in N(i)  \text{ and } \epsilon \cdot u_i(s) \geq w_{ij}/2}
$
for all $i \in S$, where $\prec$ is some total ordering on the nodes in $S$. That is, in case both $u_i(s) \ge w_{ij}/2$ and $u_j(s) \geq w_{ij}/2$ we assign edge $\{i,j\}$ to the node which is lower in the ordering $\prec$. Note that $|E_i| \leq \Delta(G)$. By definition of the set $E_i$, we now have that
\begin{align*}
\sum_{\{i,j\} \in E(s^*) : \{i,j\} \cap S \neq \emptyset \text{ and } w_{ij} > 0} w_{ij} & \leq 2 \epsilon \sum_{i \in S} \sum_{\{i,j\} \in E_i } u_i(s)  
\leq 2\epsilon \Delta(G)\sum_{i \in S}  u_i(s)  
 = 2\epsilon \Delta(G)\sum_{i \in V}  u_i(s),
\end{align*}
where the last equality holds because $u_i(s) = 0 $ for all $i \in T = V \setminus S$.  
\end{proof}

We now use this result to bound the $(\epsilon, k)$-Price of Anarchy for random graphs. Note that by exploiting the topological bound of Theorem~\ref{lem:max_degree} it suffices to bound the maximum degree of the corresponding random graph. 
The maximum degree of random graphs drawn according to the Erd\H{o}s-R\'enyi random graph model is well understood; see, e.g., Frieze and Karonski \cite{Frieze2015}.

In particular, for dense random graphs with constant $p = d \in (0,1)$, the maximum degree of a random graph satisfies $\Delta(G) \sim \Theta(n)$ (see, e.g., \cite[Chapter 3]{Frieze2015}). So for these graphs the $(\epsilon, k)$-Price of Anarchy still grows like $\Omega(\epsilon n)$ (as in the worst case).

In contrast, we obtain an improved bound for sparse random graphs. 

\begin{theorem}\label{thm:max_degree}
Let $\epsilon \geq 1$, $k \geq 2$ and $d > 0$ be constants.
Let $(c_n)_{n \in \N}$ be a sequence of integers with $c_n \geq 3$ for all $n$. 
Let $\mathcal{G}_{G_n}(c_n)$ be the set of all coordination games $\Gamma = (G_n, c_n, (S_i), \vec{1}, w,\mathbf{0})$ on graph $G_n \sim G(n,d/n)$ with $c_n$ colors, equal-split distribution rule and no individual preferences. 
Then 
$$
(\epsilon,k)\text{-PoA}(\mathcal{G}_{G_n}(c_n)) = \Theta\left(\frac{\epsilon\ln(n)}{\ln\ln(n)}\right).
$$
\end{theorem}

\begin{proof}
The bounds follow directly from Theorem \ref{lem:max_degree} and the fact that for a random graph $G_n \sim G(n,d/n)$ with $p = d/n$, we have $\Delta(G_n) \approx \mathcal{O}(\ln(n)/\ln\ln(n))$ with high probability (see, e.g., \cite[Chapter 3]{Frieze2015}). 
\end{proof}

If, in addition, the strategy sets are drawn according to a sequence of distributions that satisfy the so-called \emph{common color property}, and all weights are equal to one (corresponding to the games studied in \cite{Apt2015}), then we can even prove that the $(\epsilon, k)$-Price of Anarchy is bounded by a constant. 

\begin{definition}[Common color property]\label{def:common}
For a sequence of integers $(c_n)_{n \in \N}$, we say that a sequence of probability distributions $(\mathcal{F}_n)_{n \in \N}$ over $2^{[c_n]} \setminus \emptyset$ satisfies the \emph{common color property} if there exists some constant $d_0 > 0$ (independent of $n$) such that for $A^1_n, A^2_n \sim \mathcal{F}_n$, $\inf_{n} \mathbb{P}( A^1_n \cap A^2_n \neq \emptyset) \geq  d_0$.
\end{definition}

Intuitively, the common color property requires that with positive probability any two players have a color in common in their strategy sets. (Note that in the deterministic setting the Price of Anarchy does not improve if all players have a color in common (see \cite{Rahn2015}).) In particular, this condition is satisfied if we draw the strategy sets uniformly at random from $2^{[c]} \setminus \emptyset$ with $d_0 = \frac12$. 

\begin{theorem}\label{thm:average_degree}
Let $\epsilon \geq 1$, $k \geq 2$ and $d > 0$ be constants. 
Let $(c_n)_{n \in \N}$ be a sequence of integers with $c_n \geq 3$ for all $n$ and let $(\mathcal{F}_n)_{n \in \N}$ be a sequence of strategy set distributions satisfying the common color property. 
Let $\mathcal{G}_{G_n, (S_i)}(c_n)$ be the set of all coordination games $\Gamma = (G_n, c_n, (S_i), \vec{1}, \vec{1},\mathbf{0})$ on graph $G_n \sim G(n,d/n)$ with $c_n$ colors, strategy set $S_i \sim \mathcal{F}_n$ for every $i$, equal-split distribution rule, unit weights and no individual preferences.
Then there exists a constant $\beta = \beta(d,\epsilon)$ such that $(\epsilon,k)\text{-PoA}(\mathcal{G}_{G_n,(S_i)}(c_n)) \leq \beta$.
\end{theorem}

The proof of Theorem~\ref{thm:average_degree} relies on the following probabilistic result regarding the maximum size of a matching in Erd\H{o}s-R\'enyi random graphs \cite{Karp1981}.

\begin{lemma}[\cite{Karp1981}]
Let $\delta > 0$ be fixed. Then there is a constant $\mu^* = \mu^*(d)$ such that
$$
\lim_{n \rightarrow \infty} \mathbb{P}_{G_n \sim G(n,d/n)} \{\left|\mu(G_n)/n - \mu^*\right| \geq \delta \} = 0,
$$
where $\mu(G_n)$ is the size of a maximum matching in $G_n$.
\end{lemma}

\begin{proof}[Proof of Theorem~\ref{thm:average_degree}]
Let $M(G_n)$ be a maximum matching in $G_n$ of size $\mu(G_n)$. For a fixed edge $e = \{i,j\} \in M(G_n)$, the probability that the strategy sets $S_i$ and $S_j$ of players $i$ and $j$ satisfy $S_i \cap S_j \neq \emptyset$ is at least $d_0$ by the common color property. (Here we implicitly use that the strategy sets are drawn independently from the graph topology.)
Combining this with the lemma above, it follows that there exists a constant $\beta_0 = \beta_0(d_0)$ such that with high probability 
there exist $\beta_0 n$ pairwise node-disjoint edges in $G_n$ for which the players corresponding to the endpoints have a common color in their strategy set. (This follows from standard Chernoff bound arguments as in the proof of Theorem \ref{thm:random_graph}.) 
As a consequence, $u_i(s) + u_j(s) \geq 1/(2\epsilon)$ for any $(\epsilon,k)$-equilibrium because otherwise player $i$ and $j$ could jointly deviate (as $k \ge 2$) to their common color. 

This implies that there exists a constant $\beta_1 = \beta_1(d_0,\epsilon)$ such that with high probability
$$
(\epsilon,k)\text{-PoA}\left(\mathcal{G}_{G_n,(S_i)}(c_n)\right) \leq \beta_1 \frac{\mu(G_n)}{n}.
$$
Finally, using again standard Chernoff bound arguments as in the proof of Theorem \ref{thm:random_graph}, it follows that $\mu(G_n)/n \leq \beta_2(d)$ with high probability. This completes the proof.
\end{proof}

The statement of Theorem~\ref{thm:average_degree} does not hold for $k = 1$. To see this, consider the uniform distribution over strategy sets $\{s_0,s_1\},\dots,\{s_0,s_n\}$. In the strategy profile where every player picks her color different from $s_0$, at most a constant number of edges will be satisfied with high probability. Thus, $(\epsilon,1)\text{-PoA} \geq \beta n$ for some $\beta$ with high probability.

\section*{Acknowledgements}
The first author thanks Remco van der Hofstad for a helpful discussion on random graph theory and, in particular, the results in \cite{Anantharam2016}.

 \bibliographystyle{plain}
 \bibliography{references}

\begin{thebibliography}{10}

\bibitem{Aland2006}
Sebastian Aland, Dominic Dumrauf, Martin Gairing, Burkhard Monien, and Florian
  Schoppmann.
\newblock Exact price of anarchy for polynomial congestion games.
\newblock In Bruno Durand and Wolfgang Thomas, editors, {\em STACS 2006}, pages
  218--229, 2006.

\bibitem{Amiet2019}
Ben Amiet, Andrea Collevecchio, and Marco Scarsini.
\newblock Pure nash equilibria and best-response dynamics in random games.
\newblock {\em CoRR}, abs/1905.10758, 2019.

\bibitem{Anantharam2016}
Venkat Anantharam and Justin Salez.
\newblock The densest subgraph problem in sparse random graphs.
\newblock {\em The Annals of Applied Probability}, 26(1):305--327, 2016.

\bibitem{Anshelevich2014}
Elliot Anshelevich and Shreyas Sekar.
\newblock Approximate equilibrium and incentivizing social coordination.
\newblock In {\em Proc.~28th~AAAI~Conf.~on Artificial Intelligence}, pages
  508--514, 2014.

\bibitem{Apt2015}
Krzysztof~R. Apt, Bart de~Keijzer, Mona Rahn, Guido Sch{\"{a}}fer, and Sunil
  Simon.
\newblock Coordination games on graphs.
\newblock {\em International Journal of Game Theory}, 46(3):851--877, 2017.

\bibitem{Barany2007}
Imre B{\'a}r{\'a}ny, Santosh Vempala, and Adrian Vetta.
\newblock Nash equilibria in random games.
\newblock {\em Random Structures \& Algorithms}, 31(4):391--405, 2007.

\bibitem{Bilo2018}
Vittorio Bil\`o, Angelo Fanelli, Michele Flammini, Gianpiero Monaco, and Luca
  Moscardelli.
\newblock Nash stable outcomes in fractional hedonic games: Existence,
  efficiency and computation.
\newblock {\em Journal of Artificial Intelligence Research}, 62:315--371, 2018.

\bibitem{Caragiannis2011}
Ioannis Caragiannis, Michele Flammini, Christos Kaklamanis, Panagiotis
  Kanellopoulos, and Luca Moscardelli.
\newblock Tight bounds for selfish and greedy load balancing.
\newblock {\em Algorithmica}, 61(3):606--637, November 2011.

\bibitem{Carosi2017}
Raffaello Carosi, Michele Flammini, and Gianpiero Monaco.
\newblock Computing approximate pure {N}ash equilibria in digraph k-coloring
  games.
\newblock In {\em Proc. 16th Conf. on Autonomous Agents and Multi Agent
  Systems}, pages 911--919, 2017.

\bibitem{Carosi2018}
Raffaello Carosi and Gianpiero Monaco.
\newblock Generalized graph k-coloring games.
\newblock In {\em Proc. 24th International Computing and Combinatorics Conf.},
  pages 268--279, 2018.

\bibitem{Chen2010}
Ho-Lin Chen, Tim Roughgarden, and Gregory Valiant.
\newblock Designing network protocols for good equilibria.
\newblock {\em SIAM Journal on Computing}, 39(5):1799--1832, 2010.

\bibitem{Christodoulou2005b}
George Christodoulou and Elias Koutsoupias.
\newblock On the price of anarchy and stability of correlated equilibria of
  linear congestion games.
\newblock In {\em Proceedings of the 13th Annual European Conference on
  Algorithms}, ESA'05, pages 59--70, 2005.

\bibitem{Christodoulou2005}
George Christodoulou and Elias Koutsoupias.
\newblock The price of anarchy of finite congestion games.
\newblock In {\em Proceedings of the Thirty-seventh Annual ACM Symposium on
  Theory of Computing}, STOC '05, pages 67--73, 2005.

\bibitem{Dreze1980}
J.~H. Dr\`eze and J.~Greenberg.
\newblock Hedonic coalitions: Optimality and stability.
\newblock {\em Econometrica}, 48(4):987--1003, 1980.

\bibitem{Feldman2015}
Michal Feldman and Ophir Friedler.
\newblock A unified framework for strong price of anarchy in clustering games.
\newblock In {\em Proc. 42nd International Colloquium on Automata, Languages,
  and Programming}, pages 601--613, 2015.

\bibitem{Frieze2015}
A.~Frieze and M.~Karo{\'n}ski.
\newblock {\em Introduction to Random Graphs}.
\newblock Introduction to Random Graphs. Cambridge University Press, 2015.

\bibitem{Gilbert1959}
E.~N. Gilbert.
\newblock Random graphs.
\newblock {\em The Annals of Mathematical Statistics}, 30(4):1141--1144, 1959.

\bibitem{Gopalakrishnan2013}
Ragavendran Gopalakrishnan, Jason~R. Marden, and Adam Wierman.
\newblock Potential games are necessary to ensure pure nash equilibria in cost
  sharing games.
\newblock {\em Mathematics of Operations Research}, 39(4):1252--1296, 2014.

\bibitem{GM09}
Laurent Gourv{\`e}s and J{\'e}r{\^o}me Monnot.
\newblock On strong equilibria in the max cut game.
\newblock In {\em Proc. 5th International Workshop on Internet and Network
  Economics}, pages 608--615, 2009.

\bibitem{Gourves2010}
Laurent Gourv{\`e}s and J{\'e}r{\^o}me Monnot.
\newblock The max k-cut game and its strong equilibria.
\newblock In {\em Proc. 7th Conf. on Theory and Appl. of Models of
  Computation}, pages 234--246, 2010.

\bibitem{Hajek1990}
Bruce~E. Hajek.
\newblock Performance of global load balancing of local adjustment.
\newblock {\em {IEEE} Trans. Information Theory}, 36(6):1398--1414, 1990.

\bibitem{Hart1989}
Sergiu Hart and Andreu Mas-Colell.
\newblock Potential, value, and consistency.
\newblock {\em Econometrica}, 57(3):589--614, 1989.

\bibitem{Hoefer2007}
Martin Hoefer.
\newblock Cost sharing and clustering under distributed competition.
\newblock {\em PhD thesis}, 2007.

\bibitem{Karp1981}
R.~M. Karp and M.~Sipser.
\newblock Maximum matching in sparse random graphs.
\newblock In {\em Proceedings of the 22nd Annual Symposium on Foundations of
  Computer Science}, SFCS '81, pages 364--375, 1981.

\bibitem{Kleer2017EC}
Pieter Kleer and Guido Sch\"{a}fer.
\newblock Potential function minimizers of combinatorial congestion games:
  Efficiency and computation.
\newblock In {\em Proceedings of the 2017 ACM Conference on Economics and
  Computation}, EC '17, pages 223--240, New York, NY, USA, 2017. ACM.

\bibitem{KleerS2019}
Pieter Kleer and Guido Sch{\"{a}}fer.
\newblock Topological price of anarchy bounds for clustering games on networks.
\newblock In {\em Web and Internet Economics - 15th International Conference,
  {WINE} 2019, New York, NY, USA, December 10-12, 2019, Proceedings}, volume
  11920 of {\em Lecture Notes in Computer Science}, pages 241--255, 2019.

\bibitem{Kleer2019tcs}
Pieter Kleer and Guido Schäfer.
\newblock Tight inefficiency bounds for perception-parameterized affine
  congestion games.
\newblock {\em Theoretical Computer Science}, 754:65--87, 2019.

\bibitem{Koutsoupias1999}
Elias Koutsoupias and Christos Papadimitriou.
\newblock Worst-case equilibria.
\newblock In {\em Proc. 16th Conf. on Theoretical Aspects of Computer Science},
  pages 404--413, 1999.

\bibitem{KPR13}
Jeremy Kun, Brian Powers, and Lev Reyzin.
\newblock Anti-coordination games and stable graph colorings.
\newblock In {\em Proc. 6th Inter.~Symp.~on Algorithmic Game Theory}, pages
  122--133, 2013.

\bibitem{Luczak1991}
Tomasz {\L}uczak.
\newblock The chromatic number of random graphs.
\newblock {\em Combinatorica}, 11(1):45--54, 1991.

\bibitem{Marden2013}
Jason~R. Marden and Adam Wierman.
\newblock Distributed welfare games.
\newblock {\em Operations Research}, 61(1):155--168, 2013.

\bibitem{Rahn2015}
Mona Rahn and Guido Sch{\"{a}}fer.
\newblock Efficient equilibria in polymatrix coordination game.
\newblock In {\em Proc. 40th Symp.~on Math.~Foundations of Computer Science},
  pages 529--541, 2015.

\bibitem{Rosenthal1973}
R.~W. Rosenthal.
\newblock A class of games possessing pure-strategy {Nash} equilibria.
\newblock {\em International Journal of Game Theory}, 2:65--67, 1973.

\bibitem{RT02}
Tim Roughgarden and \'{E}va Tardos.
\newblock How bad is selfish routing?
\newblock {\em Journal of the ACM}, 49(2):236--259, 2002.

\bibitem{Roughgarden2010}
Gregory Valiant and Tim Roughgarden.
\newblock Braess's paradox in large random graphs.
\newblock {\em Random Structures and Algorithms}, 37(4):495--515, 2010.

\end{thebibliography}

\newpage
\appendix

\section{On possible extensions of coordination games}
\label{app:extensions}

We discuss some possible (natural) extensions of the coordination game model introduced in Section \ref{sec:pre}. However, we show that the results obtained in Section \ref{sec:poa} and/or Section \ref{sec:existence} no longer hold for these extensions.

\subsection{Global distribution rules}
A first natural generalization would be to look at more `global' distribution rules. However, already for slight generalizations of local distribution rules, it can be shown that there exist distribution rules that do not correspond to a generalized weighted Shapley distribution rule, but still guarantee the existence of a pure Nash equilibrium. 

For example consider \emph{edge-based distribution rules} defined by a function $g_e : N \times 2^N \rightarrow \R$ for every $e = \{a,b\} \in E$ determining shares $g_e(i,S) = \alpha_{i,e,S}$ so that if players $a$ and $b$ play the same color (we define $g_e(i,S) = 0$ if $e \notin S$), and $\{a,b\} \subseteq S \subseteq N$ is the set of all players that also play that common color, then player $i \in S$ receives a share of
$$
\frac{\alpha_{i,e,S}}{\sum_{j \in S} \alpha_{j,e,S}} w_{ab}
$$
of the edge weight $w_{ab}$, that is, his utility in strategy profile $s$ is (remember that $C_{s_i}$ is the set of all players choosing color $k = s_i$)
$$
u_i(s) = q_i(s_i) + \sum_{ e \in E : e \subseteq C_{s_i} } \frac{\alpha_{i,e,C_{s_i}}}{\sum_{j \in C_{s_i}} \alpha_{j,e,C_{s_i}}} w_e.
$$
For example this captures the case of egalitarian sharing in which every edge weight is shared equally between all players choosing the same color (if $\alpha_{i,e,S} = 1$ for all $i \in S$). 

We show that there exists an edge-based distribution rule, not corresponding to a generalized weighted Shapley value, that guarantees the existence of a pure Nash equilibrium. Consider the graph $G = (V,E)$ with $V = \{1,2,3\}$ and $E = \{1,2\}$.  We define $\alpha_{3,\{1,2\},S} = 0$ for $S = \{1,2,3\}$, that is, player 3 never gets a share of edge $\{1,2\}$. Moreover, we define $\alpha_{1,\{1,2\},\{1,2\}} = \alpha_{2,\{1,2\},\{1,2\}} = 1$, but $\alpha_{1,\{1,2\},\{1,2,3\}} = 1$ and $\alpha_{1,\{1,2\},\{1,2,3\}} = 3$. That is, if players $1$ and $2$ play a common color, and player $3$ a different color, then the edge-weight is split evenly, whereas if player $3$ also plays the same color, then the shares are $1/4w_{12}$ for player $1$, and $3/4w_{12}$ for player $2$. Roughly speaking, although player $3$ never receives a share from edge $w_{12}$, he does in fact influences how the edge-weight is split between players $1$ and $2$.

For any fixed number of colors, and sets of individual preferences, and any weight $w_{12}$, it can be shown that a pure Nash equilibrium always exists. Even stronger, it can be shown that any better-response sequence converges to a pure Nash equilibrium. Moreover, a similar example can be embedded in an arbitrary graph topology. The details of these observations are left to the reader.


\subsection{Hypergraph coordination games}
Another natural extension would be to consider hypergraph coordination games where edges can have size larger than two as well. However, it can be shown that for any hypergraph $G = (V,\mathcal{E})$, in which there is at least one pair of edges that share strictly more than one element,  there exists a local distribution rule $\alpha$ (describing for every edge how the edge-weight is split among the players part of that edge, in case all choose the same color) so that for all coordination games in $\mathcal{G}(G,c,\alpha)$ a pure Nash equilibrium is guaranteed to exist, but $\alpha$ does not correspond to a generalized weighted Shapley distribution rule.

Moreover, the Price of Anarchy also immediately becomes unbounded already on instances with one hyper-edge of size three.

\subsection{Color-dependent edge-weights}
Another possible extension would be to introduce color-dependent edge-weights, so that the edge-weights split between two players might can differ depending on the common color that they have. The characterizations in Theorem \ref{thm:best_response_2} and \ref{thm:pure_nash} still hold (the results obtained there are even stronger, since we can obtain the characterization already in the special case that the edge-weights are actually color-independent). However, the Price of Anarchy  becomes unbounded already on a graph with one edge.

\section{Resource Allocation Games}\label{sec:app_resource}

A resource allocation game \cite{Gopalakrishnan2013,Marden2013} 
$
\mathcal{G} = (N,R,(S_i)_{i \in N},(W_r)_{r \in R},(f_r)_{r \in R})
$
is given by a set $N = \{1,\dots,n\}$ of players, a set $R = \{1,\dots,m\}$ of resources, and \emph{strategy sets} $S_i \subseteq 2^R$  for players $i \in N$. Moreover, $W_r :2^N \rightarrow \R$ denotes the welfare function for resource $r \in R$ and $f_r$ the distribution rule of resource $r \in R$. A \emph{distribution rule} $f^W : N \times 2^N$ for welfare function $W$ is mapping 
with $f(i,S) = 0$ if $i \notin S \subseteq N$. We assume that the $f_r^W$ are \emph{efficient}, meaning that $\sum_{i \in S} f_r(i,S) = W_r(S)$ for all $S \subseteq N$. 

For a given strategy profile $s = (s_1,\dots,s_n) \in \times_i S_i$ the utility (pay-off) of player $i$ is defined as
$$
u_i(s) = \sum_{r \in s_i} f_r(i,N_r(s))
$$ 
with $N_r(s) = \{i \in N : r \in s_j\}$ the set of players using resource $r$ in profile $s$.

A well-known result from cooperative game theory states that for any fixed welfare function $W$, there exist real numbers $(\beta_T^W)_{T \subseteq N}$ such that 
$$
W(S) = \sum_{T \subseteq N} \beta_T^W g_T(S)
$$
where, for $T \subseteq 2^N$, $g_T : 2^N \rightarrow \R$ is the welfare functions given by $g_T(S) = 1$ if $T \subseteq S$ and zero otherwise. A distribution rule $f^W$ is said to have a \emph{base decomposition} \cite{Gopalakrishnan2013} if it can be written as
$$
f^W(i,S) = \sum_{T \subseteq N} \beta_T^W f^T(i,S)
$$
where $f^T(i,S)$ is given by $f^T(i,S) = 0$ if $T \nsubseteq S$, and, if $T \subseteq S$,  $f^T(i,S) = \omega^T$ where $\omega(i) = 0$ if $i \notin S$, and $\omega(i) > 0$ for at least one $i \in S$. This is equivalent to saying that the distribution rules $f^T$ for the welfare functions $g_T$ is a generalized weighted Shapley distribution rule \cite{Gopalakrishnan2013}.
\medskip

\subsection{Clustering games as resource allocation games}
For a fixed graph $G = (V,E_c \cup E_a)$, distribution rule $\alpha$, and $c \in \N$, any game in $\mathcal{G}(G,c,\alpha)$ can be modeled as a resource allocation game. That is, for every $\Gamma \in \mathcal{G}$, there exists a resource allocation game $\Psi = (N,R,(S_i)_{i \in N},W,f^W)$ with a one-to-one correspondence between the strategy profiles of $\Gamma$ and $\Psi$ that preserves improving moves. Here, every resource is equipped with welfare function 
\begin{equation}\label{eq:app_welfare}
W(S) = \sum_{T \subseteq \{V,E\}} \beta^W_T g_T
\end{equation}
where $\beta^W_T = 1$ if $T \in V$ or $T \in E$ with $\tau(T) = 1$, and $\beta^W_T = -1$ if $ T \in E$ with $\tau(T) = 0$. Note that the welfare function $W$ is independent of $w$ and $q$. Moreover, the distribution rule $f^W$ has a base decomposition given by $\alpha$. That is, the value $\beta_T^W$ for $T = \{i\}$ with $i \in V$ is always given to player $i$, and for $T \in E$, the corresponding weight $\beta_W^T \in \{-1,1\}$ is split among the players in $T$ according to $\alpha$ (note that this yields an efficient distribution rule).

The modeling of a clustering game as a resource allocation game is done by including many copies of a single resource, a technique also used by Gopalakrishnan et al. \cite{Gopalakrishnan2013}. 
The details of this procedure are not hard to derive and left to the reader at this point.

\subsection{Interpretation of Theorem \ref{thm:pure_nash}.}\label{sec:comparison} 

Gopalakrishnan et al. \cite{Gopalakrishnan2013} show the impressive result that, for any fixed welfare function $W$, if a distribution rule $f^W$ guarantees the existence of a pure Nash equilibrium in \emph{any} resource allocation game $(N,R,(S_i)_{i \in N},W,f^W)$, for arbitrary $N, R,$ and $(S_i)_{i \in N}$), then the distribution rule $f^W$ must be a generalized weighted Shapley distribution rule. (We refer the reader to \cite{Gopalakrishnan2013} for the formal definition of generalized weighted Shapley distribution rules for general resource allocation games. Moreover, any generalized weighted Shapley distribution rule guarantees pure Nash equilibrium existence \cite{Gopalakrishnan2013}.)

Roughly speaking, they first show that if an equilibrium is guaranteed to exist in any game where resources are equipped with welfare function $W$, then the distribution rule $f^W$ must have a base-decomposition (as introduced above). They then continue by showing that generalized weighted Shapley distribution rules (which are base-decomposable by definition) are the only ones guaranteeing existence among all base-decomposable distribution rules.\medskip

In Theorem \ref{thm:pure_nash} we essentially give an alternative, but also stronger, proof for this final step of the proof of Gopalakrishnan et al. \cite{Gopalakrishnan2013}, in the (very) special case where $W$ is of the form (\ref{eq:app_welfare}) and $\beta_W^T > 0$ for all $T \in \{V,E\}$. That is, we show that if a pure Nash equilibrium is always guaranteed to exist in a coordination game with individual preferences, where there are three common strategies (colors), then the distribution rule must be a generalized weighted Shapley distribution rule. This then implies the result of Gopalakrishnan et al. \cite{Gopalakrishnan2013}, since coordination games with individual preferences essentially form a subclass of all resource allocation games where resources are equipped with $W$ (using the modeling of clustering games as resource allocation games mentioned before).

However, Example \ref{exmp:counter_clustering} below illustrates that, in general, this is not true if $\beta_W^T < 0$ for some $T \in E$. That is, if certain coefficients $\beta_W^T$ are negative, then in general it does not suffice to focus on the subclass of corresponding clustering games, in order to derive that $\alpha$ must be a generalized weighted Shapley distribution rule. In this case, one has to make use of more complex resource allocation games, i.e., more complex than clustering games with individual preferences, in order to guarantee that $\alpha$ is a generalized weighted Shapley distribution rule (the resource allocation games used by Gopalakrishnan et al. \cite{Gopalakrishnan2013} for this final step are indeed more complex than clustering games in this case).

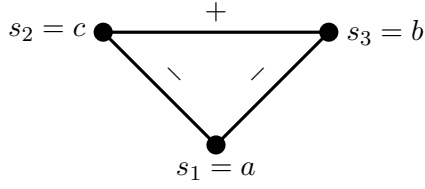
\begin{figure}[t]
\centering
\begin{tikzpicture}[scale=3]
\coordinate (1) at (0,0); 
\coordinate (2) at (-0.5,0.5);
\coordinate (3) at (0.5,0.5);

\node at (1) [circle,scale=0.7,fill=black] {};
\node (a1) [below=0.1cm of 1]  {$s_1 = a$};
\node at (2) [circle,scale=0.7,fill=black] {};
\node (a2) [left=0.1cm of 2]  {$s_2 = c$};
\node at (3) [circle,scale=0.7,fill=black] {};
\node (m1) [right=0.1cm of 3]  {$s_3 = b$};

\path[every node/.style={sloped,anchor=south,auto=false}]
(1) edge[-,very thick] node {$-$} (2) 
(2) edge[-,very thick] node {$+$} (3)
(3) edge[-,very thick] node {$-$} (1);  
\end{tikzpicture}
\caption{Counter-example for Theorem \ref{thm:poa} in the case of a clustering game with anti-coordination edges.}
\label{fig:counter_clustering}
\end{figure}

\begin{example}\label{exmp:counter_clustering}
Consider the instance in Figure \ref{fig:counter_clustering}, and let $\alpha$ be some arbitrary local distribution rule. 
Fix arbitrary weights $w_{12}, w_{23}$ and $w_{31}$ and individual preferences $q^i_l$ for $i = 1,2,3$ and $l \in  C = \{1,\dots,c'\}$. (We use $c'$ here to denote the number of strategies in the common strategy set instead of $c$.) We claim that a pure Nash equilibrium always exists.

Consider the strategy profile in Figure \ref{fig:counter_clustering} and assume without loss of generality that $a$ is the color for which player $1$ his individual preference is maximal, i.e., $a = \arg\max \{ q_1(l) : l \in C\}$. If there is some profile $s' = (a,c,b)$ with $a \neq b,c$ (but possibly $b = c$) where $c$ and $b$ are best-responses for resp. players $2$ and $3$, then we find a pure Nash equilibrium, by definition of $a$. 

It now suffices to show that for any profile of the form $(a,c,b)$ where either player $2$ or $3$ has a best-response to $a$, we can always perform a sequence of best-response moves that end up in a pure Nash equilibrium. Assume without loss of generality, that $a$ is a best-response for player $2$ in the profile $(a,c,b)$.
This in particular implies that 
$$
a = \arg\max \{ q_2(l) : l \in C\}.
$$ 
We now consider player $3$ in the profile $(a,a,b)$.
\begin{enumerate}[1)]
\item \textbf{Player $3$ only  has $a$ as best-response.} Then we let player $3$ switch to $a$ to get the profile $(a,a,a)$. Note that $a$ is still a best-response for player $2$ as well, since $w_{23}$ is non-negative and edge $\{2,3\}$ is now satisfied as well. To summarize, both players $2$ and $3$ are playing a best-response in the profile $(a,a,a)$. Now, if player $1$ has a best-response different from $a$, say $c$, then in particular $a$ remains a best-response for both players $2$ and $3$ in the profile $(c,a,a,)$, since edges $\{1,2\}$ and $\{1,3\}$ are anti-coordination edges and their weights are non-negative. That is, $(c,a,a)$ is a pure Nash equilibrium.
\item \textbf{Player $3$ has only $b$ as best-response.} Then both players $2$ and $3$ are playing a best-response in the profile $(a,a,b)$. Now suppose player $1$ has a different best-response than $a$.
\begin{enumerate}[i)]
\item \textbf{Player $1$ has a best-response to some color $d \neq b$.} Then $a$ is still a best-response for player $2$ in the profile $(d,a,b)$. If player $3$ now has a best-response other than $b$, then it must be $a$, otherwise he would have had a response better than $b$ in the profile $(a,a,b)$ as well. Clearly, in the profile $(a,a,d)$ both players $2$ and $3$ are playing a best-response. If player $1$ still has a best-response, it must be $b$ (otherwise he would have had a different best-response than $d$ before). The profile $(a,a,b)$ is a pure Nash equilibrium.
\item \textbf{Player $1$ only has a best-response to $b$.} Then $a$ is still a best-response for player $2$ in $(b,a,b)$. Suppose that player $3$ has a best-response other than $b$. If $a$ is a best-response for player $3$, then we reach the pure Nash equilibrium $(b,a,a)$, since player $2$ clearly plays a best-response, and player $1$ cannot have a better response, otherwise deviating to $b$ in the profile $(a,a,b)$ was not a best-response. 

\ \ \ \ \ Therefore, suppose player $3$ has a best-response different from $a$, say $e$. Then $b$ is still a best-response for player $1$. If player $2$ has a better response than $a$, then it must be $e$, otherwise $a$ would not have been a best-response in the initial profile. Clearly player $3$ plays a best-response in the profile $(b,e,e)$. If player $1$ still has a better response than $b$, then it must be $a$, otherwise $b$ would not have been a best-response in the profile $(b,a,b)$. The resulting profile $(a,e,e)$ is a pure Nash equilibrium, since player $3$ cannot play a better response, otherwise $e$ would not have been a best-response in the profile $(a,b,b)$.
\end{enumerate}
\item \textbf{Player $3$ has $c \neq a,b$ as best-response.} 
Player $2$ now cannot have a best-response to some color $f \neq c$, otherwise $a$ would not have been a best-response in the initial profile $(a,b,b)$. Therefore, suppose that $c$ is a best-response. Then the resulting profile $(a,c,c)$ is a pure Nash equilibrium, since player $1$ has maximum possible utility, and player $3$ clearly has no better response than $c$, otherwise $c$ would not have been a best-response in $(a,a,b)$. We can now assume to be in the profile $(a,a,c)$ in which players $2$ and $3$ play a best-response. 
\begin{enumerate}[i)]
\item \textbf{Player $1$ has a best-response to $d \neq c$.} Then either the resulting profile $(d,a,c)$ is a pure Nash equilibrium, or player $3$ still has a best-response to $a$, but then the resulting profile $(d,a,a)$ is a pure Nash equilibrium.
\item \textbf{Player $1$ has a best-response to $c$.} Then player $2$ still has $a$ as best-response. Suppose that player $3$ now has a better response. If it it $a$, then the resulting profile $(a,a,c)$ is a pure Nash equilibrium. Therefore, suppose that player $3$ has a better response to some color $g \neq a$. Then $c$ is still a best-response for player $1$. If $a$ is also still best-response for player $2$, then $(c,a,g)$ is a pure Nash equilibrium. Therefore, suppose that player $2$ has a better response. Then this must be $g$ (similar reasoning as before). Clearly, in the profile $(c,g,g)$ player $3$ is still playing a best-response. Suppose that player $2$ still has a better-response, then this must be $a$. The profile $(a,g,g)$ is a pure Nash equilibrium.
\end{enumerate}
\end{enumerate}
\end{example}

\end{document}